\documentclass{article}

\usepackage{amsthm}
\usepackage{amssymb,amsmath}
\usepackage{amsfonts}
\usepackage[english]{babel}
\usepackage[utf8]{inputenc} 
\usepackage[T1]{fontenc}
\usepackage{fancyhdr}
\usepackage{yfonts}
\usepackage{mathrsfs}
\usepackage[dvips]{graphicx}
\usepackage{enumerate}
\usepackage{xspace}
\usepackage{cite}
\usepackage{psfrag}
\usepackage{url}
\usepackage[dvipsnames]{xcolor}
\usepackage{xcolor}
\usepackage{lipsum}
\usepackage{verbatim}
\usepackage{dsfont}
\usepackage{tikz}

\usepackage{hyperref}

\newtheorem{theorem}{Theorem}
\newtheorem{lemma}{Lemma}

\newtheorem{proposition}{Proposition}
\newtheorem{corollary}{Corollary}

\newtheorem{assumption}{Assumption}
\newtheorem{remark}{Remark}

\newcommand{\prt}[1]{\left( #1 \right)}
\newcommand{\brk}[1]{\left[ #1 \right]}
\newcommand{\brc}[1]{\left\{ #1 \right\}}
\newcommand{\quotes}[1]{``#1"}
\newcommand{\abs}[1]{\left|#1\right|}

\def\R{\mathbb{R}}
\def\N{\mathbb{N}}
\def\T{\mathcal{T}}
\def\E{\mathbb{E}}
\def\P{\mathbb{P}}
\def\DD{\mathcal{D}}
\def\NN{\mathcal{N}}

\providecommand{\1}{\textbf{1}}
\newcommand{\nn}{\bar n}

\newcommand{\SqM}{\bar{x}^2}
\newcommand{\MoS}{\overline{x^2}}
\newcommand{\Var}[1]{\Delta{#1}}

\newcommand{\Gos}[1]{{\mathrm{Gos}_{#1}}}
\newcommand{\Arr}[1]{{\mathrm{Arr}_{#1}}}
\newcommand{\Dep}[1]{{\mathrm{Dep}_{#1}}}
\newcommand{\Rep}[1]{{\mathrm{Rep}_{#1}}}
\newcommand{\Gn}{{\Gos{n}}}
\newcommand{\An}{{\Arr{n}}}
\newcommand{\Dn}{{\Dep{n}}}
\newcommand{\Rn}{{\Rep{n}}}

\newcommand{\source}[1]{{s_{#1}}}
\newcommand{\arrival}[1]{{a_{#1}}}
\newcommand{\evorder}{\mathcal{E}}
\newcommand{\evorderset}{\mathscr{E}}
\newcommand{\ev}{N_\evorder}

\usepackage{color}

\providecommand{\com[2]}{\begin{tt}[#1: #2]\end{tt}}
\providecommand{\norm[1]}{\left\|#1\right\|}

\usepackage{graphicx}          

\begin{document}

\title{Evolution of Descriptors in Open Multi-Agent Systems
for Gossip Interactions and Beyond\thanks{This work is supported by the ``\textit{SIDDARTA}'' ARC at UCLouvain 
and by NRDI (National Research, Development and Innovation Office) grant KKP 137490. Corresponding Author J.M. Hendrickx}}

\author{Charles~Monnoyer~de~Galland\thanks{SPW, Belgium, email: {charles.monnoyer@spw.wallonie.be}} \and Samuel~Martin\thanks{Cabestan, Grenoble, France, email: {samuel.martinsa@gmail.com}} \and Bal{\'a}zs Gerencs{\'e}r\thanks{HUN-REN Alfr\'ed R\'enyi Institute of Mathematics, Budapest, Hungary, E\"otv\"os Lor\'and University, Budapest, Hungary, email: {gerencser.balazs@renyi.hu}} \and Julien~M.~Hendrickx\thanks{ICTEAM Institute, UCLouvain, Belgium, email:{julien.hendrickx@uclouvain.be}}}

\maketitle

\begin{abstract}                          
We consider \emph{open multi-agent systems}, which are systems subject to frequent arrivals and departures of agents while the studied process takes place.
We study the behavior of all-to-all pairwise gossip interactions in such open systems.
Arrivals and departures of agents imply that the composition and size of the system evolve with time, and in particular prevent convergence.
We describe the expected behavior of the system by showing that the evolution of scale-independent quantities can be characterized exactly by fixed-size linear dynamical systems that correspond to Markov processes with states of potentially unbounded dimension. For this purpose we develop a general methodology that could be applicable to larger classes of models.
We apply this approach to characterize the evolution of the two first moments (and thus also of a variance-like measure of disagreement) for open systems of fixed and variable size.
\end{abstract}

\section{Introduction}
\label{sec:introduction}

Flexibility and scalability are among the most cited and desired features of multi-agent systems.
Real-life examples include flock of birds \cite{delporte2006birds}, ad-hoc networks of mobile-devices \cite{ApMAS:WSN}, the Internet, or vehicle coordination \cite{ApMas:Consensus_based_formation_control:CDC2019}, for which potential agent failures or new agent arrivals are expected to be handled by the system.
However, the classical models used to study multi-agent systems assume that their composition, as complex as it can be, remains unchanged over time, allowing for asymptotic results such as convergence or synchronization.

If arrivals and departures of agents are sufficiently rare as compared to the time-scale of the process that is studied by the system, this apparent contradiction may be justified, as the system is expected to be able to incorporate the effect of such event before the next one occurs.
Nevertheless, this may not be the case for certain large systems, as the arrival or departure probability of an agent and the characteristic length of a process both typically grow with the number of agents.
Living systems with birth processes \cite{mironchenko2023live} or human societies are examples where growth is proportional to the size.
Extreme environments where communication within the system may be difficult or infrequent can also lead to slow convergence rates naturally comparable to agent failure rates, and thus also challenge this assumption.
Moreover, some systems are inherently open, such as \textit{e.g.}, a stretch of road where vehicles keep joining and leaving forever.

We consider here \emph{open multi-agent systems}, which are subject to permanent arrivals and departures of agents during the execution of the process that is considered, see \textit{e.g.}, Figure~\ref{fig:fixed_system_size_with_replacement_small}.
These arrivals and departures result in new challenges for the design and analysis of such systems.
\begin{figure}[!htbp]
\centering
    \includegraphics[width=0.75\textwidth,clip = true, trim=3.5cm 9.5cm 3.5cm 10cm,keepaspectratio]{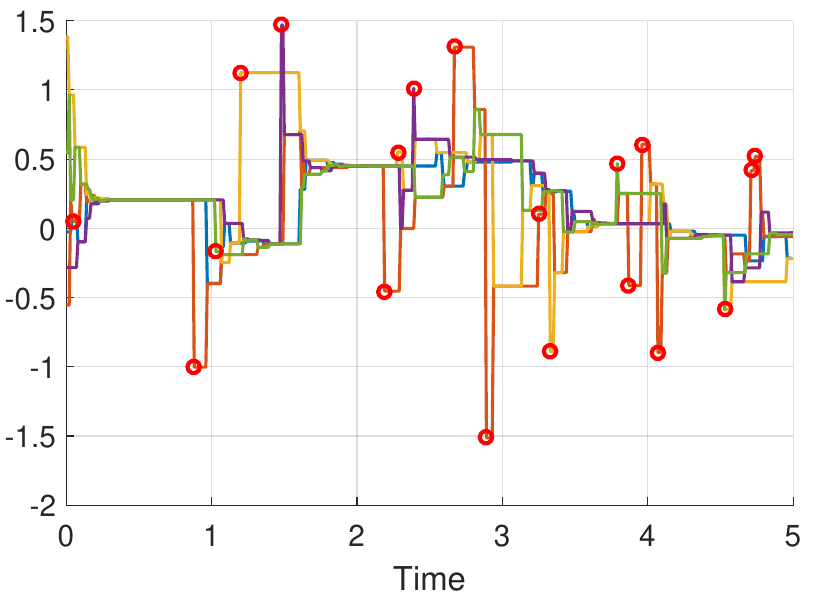}
    \caption{Example of dynamics of an open multi-agent system of $4$ agents subject to random pairwise average gossips \cite{boyd2006randomized} and agent replacements, the latter happening once every nine gossip steps on average (see Section~\ref{sec:App:FixedSize} for a precise description of the system). Each line corresponds to the value held by an agent, and each red circle highlights a replacement. The repeated replacements prevent convergence to consensus.}
    \label{fig:fixed_system_size_with_replacement_small}
\end{figure}
Firstly, every arrival and departure changes the system dimension, posing new challenges for its analysis.
Moreover, the system states suffer from repeated potentially important changes, preventing asymptotic convergence to some specific state (see \textit{e.g.},  Figure~\ref{fig:fixed_system_size_with_replacement_small}).
Instead, the general behavior of the system is expected to reach some form of steady state, which can be described by well-chosen quantities.

Secondly, arrivals and departures significantly impact the design of decentralized algorithms.
On the one hand, departures often imply losses of information that could be needed depending on the nature of the problem.
On the other hand, the desired output of the algorithm can be defined by the values held by the agents presently in the system, and thus vary over time: it can then become necessary to eliminate outdated information.
Algorithms in open systems must thus be robust to arrivals and departures, and able to handle potentially variable objectives.
Such algorithms were already explored for instance for the MAX-consensus problem in \cite{OMAS:MAX} or the median consensus problem in \cite{OMAS:Median_consensus:CDC2019}, both subject to arrivals and departures.
Moreover, algorithms designed for open systems cannot be expected to provide \quotes{exact} results.
Fundamental limitations on the performance of averaging algorithms in open systems were for instance exposed in \cite{OMAS:CDC2019:FPL_intrAVG,OMAS:ARXIV_ITAC:2020}.
Hence, it may be preferable to maintain an approximate answer robust to perturbations rather than ensuring asymptotic convergence to an exact answer if the composition were to remain constant.

We refer the interested reader to \cite{de2022open} for a more complete discussion of the challenges related to the open character of multi-agent systems.

\subsection{Contribution and preliminary results}
\label{sec:introduction:contribution}

We focus on the analysis of open multi-agent systems subject to random pairwise gossip interactions \cite{boyd2006randomized}, with the goal of developing an approach applicable to general open systems.
We consider all-to-all (possible) pairwise interactions, focusing on systems where departures and arrivals take place at random times, see Section~\ref{sec:sys-description} for a complete definition.

We show that the evolution of open systems can be analyzed in terms of \quotes{scale-independent} quantities, which we call \emph{descriptors}.
We provide a general methodology to analyze their evolution in Section~\ref{sec:Generalization}, following the observation that our system is a particular case of a general model involving Markov jump processes, which provides a useful framework, however in a challenging way given that both the jump rates and the system state dimension can go unbounded.

We apply this general approach to our system for two particular cases.
In Section~\ref{sec:App:FixedSize} we obtain an exact result for the evolution of the descriptors for fixed-size systems subject to only gossip interactions and replacements.
In Section~\ref{sec:App:VarSize} we focus on varying-size systems where arrivals and departures are decoupled; for simplicity and concision 
we directly work on a variance-like measure of disagreement and obtain an upper bound on its evolution.

Preliminary versions of these results, presented in \cite{HendrickxMartin2016allerton,HendrickxMartin2017CDC}, focused on specific discrete-time systems subject to deterministic or probabilistic events.
By contrast, we now provide continuous-time tools that formalize our methodology and allow for more general analyses, relying on simpler and more concise proofs. 
The system description and a part of its analysis is also present  
in Chapter 4 of the PhD thesis \cite{de2022open}. The current paper contains a full and rigorous analytical treatment of the general methodology in 
Section~\ref{sec:Generalization}, together with an explanation of the necessary technical assumptions.

\subsection{Other works on open multi-agent systems}
\label{sec:introduction:stateOfArt}

Consensus has been considered in different works, \textit{e.g.}, in \cite{OMAS:Varmaetal:MajorityActionPreservation,OMAS:ITAC:OpenDynamicConsensus:Franseschelli-Frasca,hadjicostis2024distributed}.
Since usual notions of consensus and stability cannot be applied in open systems, those works define alternatives for these criteria, respectively in systems subject to activations and deactivations of nodes in the former, and in systems subject to arrivals and departures in the same sense as our work for the latter. 
Our approach differs however from those, as it aims at describing the \quotes{expected} behavior to which the system is expected to converge, by following relevant scale-independent quantities that capture its \quotes{averaged} behavior.

More generally, the possibility for agents to join and leave the system has been recognized in computer science, \textit{e.g.}, through specific architectures designed to deploy large-scale open multi-agent systems \cite{carrascosa2009service}, or with algorithms designed to maintain connectivity into P2P networks subject to arrivals and departures \cite{ApOMAS:OpenP2P}, see also \cite{studli2023plug}.
Similar frameworks were also considered \textit{e.g.} in \cite{ApOMAS:VTL} for dealing with autonomous cars joining cross-sections, or in the context of trust and reputation where the reliability of joining agents must be evaluated in \cite{huynh2006integrated}.

Self-stabilizing protocols in \cite{angluin2008self,delporte2006birds}, which could also be modelled as open systems, considered the possibility for agents to undergo temporary or permanent failures: the objective is then to ensure asymptotic stability on the desired answer if the system \quotes{closes}.
Empirical and simulation-based analyses have also been conducted in \cite{torok2013opinions,iniguez2014modeling} and in \cite{OMAS:sociophysics} respectively for opinion dynamics and social phenomena subject to arrivals and departures \cite{vizuete2023lyapunov}.

Openness is also getting attention in the context of decentralized optimization,
\textit{e.g.}, in \cite{DO:online-varyingFunctions} where varying objective functions are considered, and in \cite{OpenDo:OpenDGDStability} where the stability of the decentralized gradient descent in open system is analyzed, or \cite{de2024random,deplano2026optimization,sawamura2024distributed}. Finallly, we note some recent works that consider games in open multi-agent systems \cite{taha2024learning,bistritz2024gamekeeper}.

\section{System Description and preliminaries}
\label{sec:sys-description}
We consider a multi-agent system whose composition evolves with time.
We use integers to label agents, and denote by $\NN(t)\subset\N$ the set of agents present in the system at time $t$ and by $n(t)= |\NN(t)|$ the number of agents.
Observe that the dimension of the system can change with time.
Each agent $i$ holds a value $x_i(t)\in\R$; we assume that the initial value held by the agents at $t=0$ is 
randomly and independently drawn from some distribution $\DD$ with mean $0$ and variance $\sigma^2$. 
Note that our results can immediately be adapted to distributions $\DD$ with arbitrary constant mean $\mu$.
Finally, we define the vector $x(t)\in\R^{n(t)}$ containing the values $x_i(t)$ of all the agents $i\in\NN(t)$ present in the system at that time.

\subsection{Evolution of the system}
\label{sec:sys-description:evolution}
The evolution of the system is characterized by four types of instantaneous modifications of both $n(t)$ and $x(t)$: \textit{gossip interactions}, \emph{arrivals}, \emph{departures} and \emph{replacements}.
Those modifications, whose effects are described below, are triggered by independent Poisson clocks with size-dependent rates $\lambda_{\Gn}$, $\lambda_{\An}$, $\lambda_{\Dn}$ and $\lambda_{\Rn}$ respectively.
In the sequel, we use $y$ and $y^+$ to denote the value of the time-varying quantity $y$ before and after its modification at the activation of a clock. To be consistent with classical c\`adl\`ag processes requiring, right-continuity, $y^+$ will be the value at the time at which the clock activates, and $y$ the limit from the left.

\begin{enumerate}[(a)]
    \item A \emph{gossip interaction} among $n>0$ agents, denoted $\Gn$, consists in selecting uniformly randomly and independently $i,j\in\NN(t)$ (with possibly $i=j$) to update their values $x_i$ and $x_j$ by performing a pairwise average $x_i^+ = x_j^+ = \frac{x_i+x_j}{2}$.
    Hence there holds $n^+=n$.
    
    \item An \emph{arrival} in a system of $n\geq 0$ agents, denoted $\An$, consists in one \quotes{new} agent $i\notin\NN(s)$, $\forall s\leq t$ joining the system, so that $\NN^+ = \NN\cup\brc{i}$, and thus $n^+ = n+1$.
    The initial value of the arriving agent is then independently and randomly drawn from the distribution $\DD$ used to initialize the system.
    
    \item A \emph{departure} among $n>0$ agents, denoted $\Dn$, consists in having a uniformly randomly chosen agent $i\in\NN(t)$ leave the system, so that $\NN^+=\NN\setminus\brc{i}$ and $n^+=n-1$.
    
    \item A \emph{replacement} among $n>0$ agents, denoted $\Rn$, consists in the simultaneous occurrence of both a departure and an arrival. 
    There thus holds $n^+=n$.
\end{enumerate}
Observe that the way we model arrivals and departures implies that agents leaving the system never join it again (unlike systems subject to temporary disconnections of agents, see \textit{e.g.}, \cite{OMAS:ImpactOfNoise_RandomConsensusAlgo:vizuete2021}). 
Moreover, we consider in this work that the departing agent is uniformly randomly  selected, but other definitions could be considered, \emph{e.g.}, \emph{adversarial} departures where no assumption can be made on that choice.

\subsection{Scale-independent quantities}
\label{sec:sys-description:scale-independent-quantities}
Because the set of agents $\NN(t)$ and its size change with time, it is a challenge to follow the actual full state $x(t)$ in our setting.
However, we show that we can study the evolution of scale-independent quantities, \textit{i.e.}, whose values do not scale with the system size, with a finite-dimensional affine dynamical system.
In particular, we consider the empirical squared mean $\SqM = \prt{\frac1n\sum_{i\in\NN}x_i}^2$ and mean of squares $\MoS = \frac1n\sum_{i\in\NN}x_i^2$, which we call \emph{descriptors} of the system, where we omit references to time to lighten the notation. These quantities also allow tracking the \emph{second order disagreement} (or \emph{disagreement} in short) $\Var(x)= \frac{1}{n} \sum_{i \in \NN} (x_i - \bar{x})^2 = \MoS - \SqM$. It could be tempting to call this the variance, however, we are measuring fluctuations within coordinates of a single realization, not how much the realizations fluctuate due to randomness. Finally, out of convenience for our analysis, we extend the definitions of $\SqM,\MoS$ and $\Var{}$ to empty systems, and set them arbitrarily to 0.

In closed systems, it is known that gossip interactions guarantee consensus (\textit{i.e.}, $\lim_{t \rightarrow \infty} \,  \max_{(i,j) \in \NN(t)^2} |x_i(t) - x_j(t)| = 0$), and thus $\lim_{t\to\infty}\Var(x(t))=0$, see \textit{e.g.}, \cite{FagnaniZampieri2007,boyd2006randomized}.
This cannot be achieved in open system as the value of any joining agent will be bounded away from the values of the agents already present in the system with non-trivial probability.
Hence, by analysing the evolution of the descriptors, we study the disagreement between the agents (\textit{i.e.}, how \quotes{far} the system is from consensus).
In particular, we will be interested in the evolution of $\E\Var(x)=\E\MoS-\E\SqM$.
Observe that one could for example have monitored $\E\bar x$ as well, which evolves following an independent one-dimensional linear system (see \textit{e.g.}, \cite{HendrickxMartin2016allerton}).
However, we omit this part of the study due to space limitations. 

Our first lemma, proved in Appendix~\ref{sec:appendix:bound_squaredmean}, bounds the descriptors in expectation.
\smallskip
\begin{lemma}
\label{lem:sys-description:bound_on_expected_quantities}
    In the setting described in Section~\ref{sec:sys-description}, and for any fixed time $t\geq0$, and $j\geq 1$, there holds
    \begin{align}
        \label{eq:lem:sys-description:bound_on_expected_quantities}
        &\E \prt{\SqM\mid n(t) = j} \leq \frac{1}{j}\sigma^2;&
        &\E \prt{\MoS\mid n(t)=j} \leq \sigma^2,
    \end{align}
    where we remind $n(t)$ is the system size at time $t$.
\end{lemma}
Observe that at the beginning of the process, there holds $\E(\SqM(0)|n(0)=j) = \frac{\sigma^2}{j}$ and $\E \prt{\MoS(0)|n(0)=j} = \sigma^2$.
Hence, the above lemma states that the expected values of the descriptors cannot exceed their initial expected value.

\subsection{Evolution of the descriptors in expectation}
\label{sec:sys-description:effect-of-events}
Let us define the vector $X = \begin{pmatrix}\SqM&\MoS\end{pmatrix}^\top$ containing the descriptors, so that $\Var(x) = \begin{pmatrix}-1&1\end{pmatrix}X = \MoS-\SqM$.
We now show that the evolution of the descriptors in expectation upon activation of the different clocks defining the evolution of the system is governed by a 2-dimensional affine system, from which we also derive the evolution of $\E\Var(x)$.
We use the same notation as in Section~\ref{sec:sys-description:evolution} in the following lemmas, whose proofs are in Appendix~\ref{sec:appendix:EffectOfEvents}.

\smallskip
\begin{lemma}[Gossip step]
\label{lem:effect_operations:gossip}
    When $\Gn$ activates there holds
    \begin{align}
        \label{eq:lem:effect_operations:gossip:X}
        &\E\prt{X^+\mid X,\Gn} = A_{\Gn}X = \begin{pmatrix} 1 &  0 \\ \frac{1}{n} & 1 - \frac{1}{n} \end{pmatrix}X.
    \end{align}
    Hence, $\E\prt{\Var(x^+)\mid\Var(x),\Gn} = \left(1- \frac{1}{n}\right) \Var(x)$.
\end{lemma}

\smallskip

\begin{lemma}[Arrival]
\label{lem:effect_operations:arrival}
    When $\An$ activates there holds
    \begin{equation}
        \label{eq:lem:effect_operations:arrival:X}
    \begin{aligned}
        \E\prt{X^+\mid X,\An} 
        &= A_{\An}X + b_{\An}\\
        &= \tfrac{n}{n+1}\begin{pmatrix} \frac{n}{n+1} &  0 \\ 0 & 1 \end{pmatrix}X + \begin{pmatrix}\frac{1}{(n+1)^2} \\ \frac{1}{n+1} \end{pmatrix}\sigma^2.
    \end{aligned}
    \end{equation}
      Moreover, there holds
$\E\prt{\Var(x^+)\mid\Var(x),\An} \leq \frac{n}{n+1}\prt{\Var(x) + \frac{\sigma^2}{n}}$ if $n\geq 1$, and is $\sigma^2$ if $n=0$.
\end{lemma}

\smallskip

\begin{lemma}[Departure]
\label{lem:effect_operations:departure}
    When $\Dn$ activates for $n>1$ there holds
    \begin{align}
        \label{eq:lem:effect_operations:departure:X}
        &\E\prt{X^+\mid X,\Dn} 
        = A_{\Dn}X
        = \begin{pmatrix} 1-\frac{1}{(n-1)^2} &  \frac{1}{(n-1)^2} \\ 0 & 1 \end{pmatrix}X.
    \end{align}
      Moreover, $\E\prt{\Var(x^+)\mid\Var(x),\Dn} = \prt{1-\frac{1}{(n-1)^2}}\Var(x)$. If $\Dn$ activates for $n=1$, we have $X^+=\begin{pmatrix}0&0\end{pmatrix}^\top$.
\end{lemma}

We remind that Lemma~\ref{lem:effect_operations:departure} holds for the departure of an agent that is uniformly randomly chosen among those
in the system before the departure.
We now consider replacements, which consists of a departures immediately followed by an arrival, and requires thus $n\geq 1$.
The next result follows from a combination of Lemmas \ref{lem:effect_operations:departure} and \ref{lem:effect_operations:arrival},
the latter applied to a system of size $n-1$ joined by a $n^{th}$ agent.

\smallskip

\begin{lemma}[Replacement]
\label{lem:effect_operations:replacement}
    When $\Rn$ activates there holds
    \begin{equation}
        \label{eq:lem:effect_operations:replacement:X}
    \begin{aligned}
        \E\prt{X^+\mid X,\Rn} 
        &= A_{\Rn}X + b_{\Rn}\\
        &= \begin{pmatrix} \frac{n-2}{n} & \frac{1}{n^2} \\ 0 & \frac{n-1}{n} \end{pmatrix}X + \begin{pmatrix}\frac{1}{n^2} \\ \frac{1}{n} \end{pmatrix}\sigma^2.
    \end{aligned}
    \end{equation}
    Moreover, there holds
    \begin{align*}
        \label{eq:lem:effect_operations:replacement:Var}
        \E\prt{\Var(x^+)\mid\Var(x),\Rn}
        \leq \frac{n^2-n-1}{n^2}\Var(x) + \frac{n^2-1}{n^3}\sigma^2.
    \end{align*}
\end{lemma}
Finally, we note that the assumption of the $x_i(0)$ being random i.i.d. with the same variance and mean as the arriving agents is only used in Lemma \ref{lem:sys-description:bound_on_expected_quantities}, and in the characterization of the disagreement $\Delta$ in Lemmas \ref{lem:effect_operations:arrival} and \ref{lem:effect_operations:replacement}. All other results remain true without this assumption. 

\section{Evolution of the descriptors as an ODE}
\label{sec:Generalization}

\subsection{Generalization of the system description}
\label{sec:Generalization:description}

Our system is a particular case of the following more general model: Let $S(t) = (n(t),x(t))$ be the state of the system, with $n(t)\in\N$ the system size at time $t$, and $x(t)\in\R^{n(t)}$ the time-varying vector containing the values of all the agents present in the system at that time.
The state $S(t)$ follows a Markov jump process, \textit{i.e.}, a continuous-time Markov process evolving with instantaneous changes at random times, called \quotes{\emph{jumps}}.

A number of possible jump types $\epsilon$ are identified, forming together the set $\Xi$. Each such jump type has an effect on $n(t)$ and $x(t)$, causing modification only for a single value of $n(t)$.
The modification of $n(t)$ is deterministic, and we can therefore define $\source\epsilon\in\N$ and $\arrival\epsilon\in\N$ such that an event of $\epsilon$ changes $n(t)$ from $\source\epsilon$ to $\arrival\epsilon$. On the other hand, $x(t)$ is randomly modified in a Markovian way. 

For each event type $\epsilon$ a Poisson clock $N_\epsilon$ with intensity $\lambda_\epsilon$ is introduced to signal the occurrences, the clocks are assumed 
independent. Note that while an event $\epsilon$ might keep occurring, it will only affect the system state $S(t)$ when $n(t) = \source{\epsilon}$.

In our system described in Section~\ref{sec:sys-description}, the event types correspond to the gossip interactions, arrivals, departures and replacements, and we have $\Xi = \bigcup_{n\in\N}\brc{\Gn,\An,\Dn,\Rn}$.
An event of $\An$ for instance happens with rate $\lambda_{\An}$ if the system size is $\source{\An}=n$, resulting in a jump of the process that changes the system size into $\arrival{\An}=n+1$ and modifies $x(t)$ by adding a new agent whose initial value is randomly drawn from the distribution $\DD$.

The notion of descriptors can also be more generally defined as {non-negative} functions of the system state $f(S(t))$, \textit{e.g.}, $\SqM$ and $\MoS$, that can be formulated as functions of both $n(t)$ and $x(t)$.
Hence, our goal is to analyze the evolution in expectation of a given non-negative measurable 
function $f:E\to\R_+^d$,
where $E$ is the state space over which $S(t)$ is defined. 
For that purpose, we provide variations of standard results on Markov processes in the following Section, that are connected with Kolmogorov equations. The main technical proofs are presented in Appendix 
\ref{sec:proof_ode}, \ref{sec:weak_dif}. 
These results will be the basis for establishing the expected evolution of our descriptors conditioned by the system size.

\subsection{General Evolution Results}

Let $(E,\mathcal M)$ be a measurable space, and $S(t) = (n(t), x(t))$ a Markov jump process as defined in Section~\ref{sec:Generalization:description} with state-space $E$. We let $f:E\to\R_+^d$ for some $d\in\N$  be a measurable non-negative function.
Moreover, for $j\in\N$, let $F_j(t) = \E\brk{f(S(t))\ |\ n(t)=j}$ and $\pi_j(t)=\P[n(t)=j]$. We further suppose that for all $\epsilon\in\Xi$ there exist $A_\epsilon\in\R^{d\times d}$ 
        \begin{equation}
\label{eq:prop:tools:DescriptorsEvolution:ode:Kolmogorov:LinearSys}
        \E\brk{f(S^+)\ |\ f(S),\epsilon} = A_\epsilon f(S),
    \end{equation}
where $S$ denotes the state before the event $\epsilon$, and $S^+$ the (random) state after the event. Again, to be consistent with classical results requiring right-continuity, we will assume that if the event $\epsilon$ takes place at time $t$, then $S^+$ is the new state at time $t$, while $S$ is the limit from the left.

Standards results concerning the evolution of $F_j(t) = \E\brk{f(S(t))\ |\ n(t)=j}$ in these conditions typically assume a finite number of states $n(t)$ or a uniform bound on all intensities $\lambda_\epsilon$, which cannot be guaranteed in our context due e.g. to the possibility of having the departure rate growing with the number of agents present in the system. Hence, we need a more general result, for which we have to remain within non-negative functions.

Our first assumption will prevent the (expected) state $x(t)$ from exploding even when it is conditioned to a series of events.

\smallskip

\begin{assumption}\label{assumption:bounded_matrix}
For every event $\epsilon\in \Xi$, there holds $||A_\epsilon||\leq 1$ for some induced norm on $\R^d$.
\end{assumption}

Our second assumption concerns the initial state. 
\smallskip

\begin{assumption}\label{assumption:bounded_initial}
There exists a $\bar F$ such that for every $j$ there holds
$$
\E ( || f(S(0)) || \mid n(0) = j) \leq \bar F.
$$
\end{assumption}
For convenience, we will assume that the matrix norm in Assumption \ref{assumption:bounded_matrix} is induced by the norm in Assumption \ref{assumption:bounded_initial}. This is not a loss of generality since all finite norms are equivalent and the value of $\bar F$ will be irrelevant. 
Finally, we have to assume a local boundedness of the overall rates: we not only require the sum of the rates of events leading to or leaving any node to be finite, but also the sum of all the rates leading to or departing from the neighbors of any node to be finite.

\smallskip

\begin{assumption}\label{assum:finite_sum_proba} For all $j \in \N$ define the total rates within distance 2 as
\begin{align*}
\mathfrak{a}_j &:= \{k \mid \exists \epsilon ~\source{\epsilon}=j, \arrival{\epsilon}=k\},\\
\mathfrak{s}_j &:= \{k \mid \exists \epsilon ~\source{\epsilon}=k, \arrival{\epsilon}=j\},\\
\lambda_{j2} &:= \sum_{\epsilon:\source{\epsilon}\in \mathfrak{a}_j\cup\{j\}} \lambda_\epsilon + \sum_{\epsilon:\arrival{\epsilon}\in \mathfrak{s}_j\cup\{j\}} \lambda_\epsilon.
\end{align*}
In particular, there holds
\begin{align*}
\lambda_{j,out}&:= \sum_{\epsilon:\source{\epsilon}=j} \lambda_\epsilon <\infty,\\
\lambda_{j,in} &:= \sum_{\epsilon:\arrival{\epsilon}=j} \lambda_\epsilon <\infty.
\end{align*}
\end{assumption}
This setup allows us to phrase one of our main theorems, expressing the expected evolution of the descriptors via an ODE system, and which is proved in Appendix \ref{sec:proof_ode}.

\smallskip

\begin{theorem}\label{thm:evol_linear}
For the class of systems described above and satisfying \eqref{eq:prop:tools:DescriptorsEvolution:ode:Kolmogorov:LinearSys}, under Assumptions \ref{assumption:bounded_matrix}, \ref{assumption:bounded_initial} and \ref{assum:finite_sum_proba}, there holds
    \begin{align}        \label{eq:prop:tools:DescriptorsEvolution:ode:Kolmogorov}
        \frac{d}{dt}(F_j(t)\pi_j(t))
        =  &\sum_{\epsilon:\arrival{\epsilon}=j}\lambda_\epsilon\pi_{\source\epsilon}(t)A_\epsilon F_{\source\epsilon}(t)\nonumber\\
        &- \sum_{\epsilon:\source{\epsilon}=j}\lambda_\epsilon \pi_j(t)F_j(t),
    \end{align}
Moreover, if \eqref{eq:prop:tools:DescriptorsEvolution:ode:Kolmogorov:LinearSys} is relaxed to a coordinate-wise inequality $\leq$ to upper bound the jump effect then \eqref{eq:prop:tools:DescriptorsEvolution:ode:Kolmogorov} holds with the inequality $\leq$ to upper bound the differential, provided that the vector norm used in Assumption \ref{assumption:bounded_initial} and inducing that in Assumption \ref{assumption:bounded_matrix} is monotonous, i.e.\ $\norm{x}\leq \norm{y}$ if $0\leq x_i \leq y_i$ for all $i$.
\end{theorem}

Observe that some of the transitions in our open multi-agent systems do not satisfy the linear relation \eqref{eq:prop:tools:DescriptorsEvolution:ode:Kolmogorov:LinearSys} but a generalized affine version 
    \begin{equation}
\label{eq:prop:tools:DescriptorsEvolution:ode:Kolmogorov:AffineSys}
         \E\brk{f(S^+)\ |\ f(S),\epsilon} = A_\epsilon f(S) + b_\epsilon,
    \end{equation}
or its inequality version, see e.g. Lemmas \ref{lem:effect_operations:arrival} and \ref{lem:effect_operations:replacement}. However, our results can directly be applied under a generalization of the boundedness Assumption \ref{assumption:bounded_matrix}.

\smallskip

\begin{assumption}\label{assumption:bounded_affine}
For every event $\epsilon\in \Xi$, there holds 
$$
\norm{\begin{bmatrix}
A_\epsilon & b_\epsilon \\
0^\top_{d} & 1
\end{bmatrix}}\leq 1
$$
for some induced norm on $\R^{d+1}$.
\end{assumption}

\smallskip

\begin{corollary}\label{cor:evol_affine}
For the class of systems described above and satisfying \eqref{eq:prop:tools:DescriptorsEvolution:ode:Kolmogorov:AffineSys}, under Assumptions \ref{assumption:bounded_initial}, \ref{assum:finite_sum_proba} and \ref{assumption:bounded_affine} there holds
    \begin{align}        \label{eq:prop:tools:DescriptorsEvolution:ode:Kolmogorov2aff}
        \frac{d}{dt}(F_j(t)\pi_j(t))
        =  &\sum_{\epsilon:\arrival{\epsilon}=j}\lambda_\epsilon\pi_{\source\epsilon}(t)\prt{A_\epsilon F_{\source\epsilon}(t)+b_\epsilon}\nonumber\\
        &- \sum_{\epsilon:\source{\epsilon}=j}\lambda_\epsilon \pi_j(t)F_j(t),
    \end{align}
Moreover, if \eqref{eq:prop:tools:DescriptorsEvolution:ode:Kolmogorov:AffineSys} is relaxed to an inequality $\leq$ to upper bound the jump effect then \eqref{eq:prop:tools:DescriptorsEvolution:ode:Kolmogorov2aff} holds with the inequality $\leq$ to upper bound the differential.
\end{corollary}
\begin{proof}
We consider a modified function $\tilde f:E\to \R_+^{d+1}$ defined by 
$\tilde f(S) = (f(S)^\top,1)^\top$. We can then write that
\begin{align*}  
    \label{eq:prop:tools:DescriptorsEvolution:ode:Kolmogorov:AffineSys}
        \E\brk{\tilde f(S^+)\ |\ f(S(t)),\epsilon} &= \begin{pmatrix} A_\epsilon f(S) + b_\epsilon\\1 \end{pmatrix}\\
        &= \begin{pmatrix} A_\epsilon & b_\epsilon \\ 0 & 1 
        \end{pmatrix} \begin{pmatrix} f(S)) \\ 1\end{pmatrix}\\ 
        &= \begin{pmatrix} A_\epsilon & b_\epsilon \\ 0 & 1 
        \end{pmatrix} \tilde f(S),
\end{align*}
leading to a linear update in expectation.

For this modification Assumption \ref{assumption:bounded_initial} holds intact, while Assumption \ref{assumption:bounded_matrix} is replaced by Assumption \ref{assumption:bounded_affine} as the updating matrix has been changed, then Theorem \ref{thm:evol_linear} can be applied to conclude.
\end{proof}

Finally, we remark a subtlety here; our results Theorem \ref{thm:evol_linear} established the differentiability of each $F_j(t)\pi_j(t)$ but not of the (infinite) sequence $((F_j(t)\pi_j(t))_{j=0,1,\dots}$ as a single object, and in particular cannot be used directly to analyze derivatives of infinite sums such as $\sum_j F_j(t)\pi_j(t)$. The following result addresses this limitation, albeit at the cost of using a weaker notion of derivative and relying on the non-negativity assumption
preventing the uncontrolled explosion of the expected value of transitions.

\smallskip

\begin{assumption}\label{assumption:bounded_expected_rate}
The system does not go out of bounds in finite time. More precisely, at any time $t$,
\[
\sup_{U\subset \N, |U|<\infty}\P(n(t)\in U) = 1.
\]
This is equivalent to $\sum_j\pi_j(t)=1$.

    Moreover, there exists a function $\bar W(t)$, integrable on all finite intervals, such that for every $t$ the following inequality holds
    \begin{equation}
    \label{eq:assumption_boundedexpectedrate}
    \sum_{\epsilon} \pi_{\source{\epsilon}}(t)\lambda_\epsilon \leq \bar W(t).    
    \end{equation}
\end{assumption}

\smallskip

\begin{proposition}\label{prop:weak_deriv}
    In the setting of Theorem \ref{thm:evol_linear} and Corollary \ref{cor:evol_affine}, further requiring Assumption \ref{assumption:bounded_expected_rate}, for any bounded sequence of numbers $(w_j)_{j\in \N}$, for any time interval $[t,t+\tau]$, there holds
\begin{align*}
\prt{\sum_j w_j F_j(t+\tau)\pi_j(t+\tau) } - \prt{\sum_j w_j F_j(t)\pi_j(t) }\\
= \int_{t}^{t+\tau} \sum_j w_j (F_j(u)\pi_j(u))' du,
\end{align*}
and in particular
$$
\E(f(S(t+\tau)) - \E(f(S(t)) = \int_{t}^{t+\tau} \sum_j (F_j(u)\pi_j(u))' du.
$$  
\end{proposition}
The proof is presented in Appendix \ref{sec:weak_dif}.

Finally, note that Assumption \ref{assumption:bounded_expected_rate} presents a practical difficulty, as it refers to the evolution of the distribution $\pi$ and not only on the structure of the Markov process, so that harder to check apriori. Therefore, we provide a powerful sufficient condition that relies only on the transition rates in the specific case of processes where $n$ can only move by 1 \emph{when it increases}, a class that encompasses all processes considered in the sequel.

\smallskip

\begin{lemma}
\label{lem:expectedrate}
Assume $\pi(0)$ is finitely supported and after possibly reordering the states, for all events $a_\epsilon-s_\epsilon \leq 1$, i.e.\ we jump up at most one. Let
\[
\mu_j = \sum_{\epsilon:s_\epsilon = j, a_\epsilon=j+1} \lambda_\epsilon,
\]
and assume $\lambda_{j,out}\leq h_j$ for an appropriate non-decreasing sequence $h_j$.
If with these variables
\[
\mu_j\frac{h_{j+1}-h_j}{h_j}\leq C,\quad\text{for all }j
\]
for some constant $C>0$, then Assumption \ref{assumption:bounded_expected_rate} is satisfied.
\end{lemma}
The proof of the lemma is presented in Appendix \ref{sec:non-explosive}.

A natural application of the lemma above is constant arrival rate for an open multi-agent system, $\mu_j\equiv \mu$, which allows exponentially growing $h_j$ in the conditions. As another special case, one could think of linearly growing $\mu_j$ and polynomial $h_j$.

\section{Applications}
\label{sec:App}
We now apply the general methodology of the previous section on our system described in Section~\ref{sec:sys-description}, \textit{i.e.}, with $\Xi = \bigcup_{n\in\N}\brc{\Gn,\An,\Dn,\Rn}$.
We restrict to two particular characterizations of that setting: 
first a fixed-size system where $n(t)=n$ for illustrative purpose, and then a more general setting where agents leave and join the system independently. 

Prior to that we first verify that our system does indeed satisfy the Assumptions of Theorem \ref{thm:evol_linear}, Proposition \ref{prop:weak_deriv} and Corollary \ref{cor:evol_affine}. 
Indeed, we are interested in following the evolution of $X(t) = (\SqM(t),\MoS(t))^\top$
which can be cast as a function of type $f:E\to \R_+^2$ required by our framework. So let us take $f(S(t)) =  X(t) = (\SqM(t),\MoS(t))^\top$.
One can directly verify from Lemmas \ref{lem:effect_operations:gossip}, \ref{lem:effect_operations:arrival}, \ref{lem:effect_operations:departure} and \ref{lem:effect_operations:replacement} that all the events in our system satisfy \eqref{eq:prop:tools:DescriptorsEvolution:ode:Kolmogorov:AffineSys} for some $A_\epsilon$ and $b_\epsilon$ satisfying the boundedness Assumption \ref{assumption:bounded_affine} with the induced $\norm{.}_1$ norm, which is monotonous. We further observe that Assumption \ref{assum:finite_sum_proba} is trivially satisfied by the structure of the system, and 
Assumption \ref{assumption:bounded_initial} only concerns the initial distribution of states having a finite second moment, and will be supposed satisfied in the sequel. 
Finally,  we will only consider constant (or zero) arrival rates $\lambda_{\An}$, and all other rates will grow at most polynomially, and in fact linearly, with $n$. Hence, Lemma \ref{lem:expectedrate} implies that 
Assumption \ref{assumption:bounded_expected_rate} is satisfied too.

\subsection{Fixed-size systems}
\label{sec:App:FixedSize}
We now consider a first simple setting in which only gossip interactions and replacements occur, so that the number of agents remains constant, \textit{i.e.}, $n(t)=n$.
They respectively take place with rates $\lambda_{\Gn}=n\lambda_g$ and $\lambda_{\Rn}=n\lambda_r$ for some $\lambda_g,\lambda_r\geq0$, and $\lambda_{\An}=\lambda_{\Dn}=0$ as no arrival nor departure ever happens.
This means that on average $n\lambda_g$ gossip interactions and $n\lambda_r$ replacements happen in the whole system per unit of time, and thus the rate of any event for a given individual agent is independent of $n$.
Moreover, the expected number of gossip interactions taking place between two replacements $\rho=\lambda_g/\lambda_r$ and the probability of a random event being a replacement $p=\frac{1}{1+\rho}$ are both independent of $n$.

\smallskip

\begin{theorem}
\label{thm:FixedSize:RandomRepl:sys-evolution}
    In a system of fixed size $n \geq 1$ 
    subject to replacements and gossips, there holds    
    \begin{align}
        \label{eq:thm:FixedSize:RandomRepl:sys-evolution}
        \frac{d}{dt} \E X(t)
        = \begin{pmatrix}
            -2\lambda_r &\frac{\lambda_r}{n}\\
            \lambda_g   & -\lambda_g-\lambda_r
        \end{pmatrix}
        \E X(t) 
        + \begin{pmatrix}\frac{\lambda_r}{n}\\\lambda_r\end{pmatrix}\sigma^2.
    \end{align}
\end{theorem}
\begin{proof}
    We consider the function $f(S(t)) = X(t)$ that computes the descriptors.
    The result follows from the application of Corollary \ref{cor:evol_affine} combined with Lemmas~\ref{lem:effect_operations:gossip} and \ref{lem:effect_operations:replacement},
    where we observe that $\pi_n(t) = 1$ and $F_n(t) = \E X(t)$ for any $t$, so that
    \begin{align*}
        \frac{d}{dt}\E X(t) 
        &= n\lambda_r \prt{A_{\Rn} \E X(t) + b_{\Rn}} + n\lambda_g A_{\Gn} \E X(t)\\
        &-n\prt{\lambda_r + \lambda_g}  \E X(t),
    \end{align*}
    where $A_{\Rn}$ and $b_{\Rn}$ come from Lemma~\ref{lem:effect_operations:replacement} and $A_{\Gn}$ from Lemma~\ref{lem:effect_operations:gossip}.
    A few algebraic steps then conclude the proof.
\end{proof}

\begin{remark}\label{rem:2dyn_vs_1dyn}
    A similar yet weaker result than that of Theorem~\ref{thm:FixedSize:RandomRepl:sys-evolution} can be obtained by applying inequality claims of Corollary \ref{cor:evol_affine} 
    directly on the disagreement, 
    yielding
    \begin{align}
        \label{eq:FixedSize:RandomRepl:BoundOnVar}
        \frac{d}{dt}\E \Var(t)
        \leq &-\prt{\lambda_g+\frac{n+1}{n}\lambda_r}\E \Var(t)
        + \frac{n^2-1}{n^2}\lambda_r\sigma^2.
    \end{align}
    By contrast, Theorem~\ref{thm:FixedSize:RandomRepl:sys-evolution} yields an exact equality for $\frac{d}{dt}\E \Var(t)$ since $\E \Var(t) = \begin{pmatrix}-1&1\end{pmatrix}\E X(t)$, which thus highlights how following two descriptors allows for a more detailed analysis in this case, even if one were only interested in the disagreement.
\end{remark}
We now briefly analyze the fixed point of \eqref{eq:thm:FixedSize:RandomRepl:sys-evolution} given by
\begin{align}
    \label{eq:FixedSize:RandomRepl:FixedPoint:SqM&MoS}
    &\E\SqM\rvert_{eq} = \frac{2+\rho}{2n(1+\rho)-\rho}\sigma^2&
    &;&
    &\E \MoS\rvert_{eq} = \frac{2n+\rho}{2n(1+\rho)-\rho}\sigma^2,
\end{align}
leading to a disagreement 
\begin{equation}
    \label{eq:asymptotic_var_constant_n}
    \E \Var(x)\rvert_{eq} 
    = \frac{1-\frac1n}{1+\rho\prt{1-\frac{1}{2n}}}\sigma^2,
\end{equation}
where we remind $\rho = \lambda_g/\lambda_r$ is the ratio between gossip and replacement rates.
As gossips become less frequent ($\rho\to0$)
the system eventually consists of agents with $n$ i.i.d. values with mean $0$ and variance $\sigma^2$, and $\E \Var(x)\rvert_{eq}\to\frac{n-1}{n}\sigma^2$. 
By contrast, as replacements become less frequent ($\rho\to\infty$) then the system starts behaving as a closed system with $\E \Var(x)\rvert_{eq}\to0$. 
Moreover, as the system size grows the expected disagreement becomes $\E\Var(x(t))|_{eq}\to\frac{\sigma^2}{1+\rho}=p\sigma^2$ (where we remind $p$ is the probability for an event to be a replacement), consistently with results derived for DT systems in a preliminary version of this work \cite{HendrickxMartin2017CDC}.
In particular, increasing the gossip rate then makes the expected asymptotic disagreement decay, as the system gets closer to consensus; conversely, increasing the replacement rate makes it increase.

A similar analysis of the convergence rate can be conducted by analyzing the eigenvalues and eigenvectors of the matrix in \eqref{eq:thm:FixedSize:RandomRepl:sys-evolution}, see e.g. \cite{HendrickxMartin2017CDC} in a discrete-time setting, but is omitted here for the sake of concision.

We illustrate the results above in Figure~\ref{fig:FixedSize:RandomRepl:Estimates&Descriptors:n50} for a system of $n=50$ agents subject to replacements and gossips in such a way that on average one event in twenty is a replacement (\textit{i.e.}, $\rho=19$). 
The agent initial values are randomly drawn from a normal distribution with zero mean and constant variance $\sigma^2=1$.
Figure~\ref{fig:FixedSize:RandomRepl:Estimates&Descriptors:n50}(a) shows how convergence is prevented by replacements even if the agents occasionally get close to
consensus with a single realization of the system.
Figures~\ref{fig:FixedSize:RandomRepl:Estimates&Descriptors:n50}(b)-(d) show that the theoretical results of Theorem~\ref{thm:FixedSize:RandomRepl:sys-evolution} and the fixed point in \eqref{eq:FixedSize:RandomRepl:FixedPoint:SqM&MoS} accurately match with the expected evolution of the descriptors from the simulation.

\begin{figure}
\centering
    \includegraphics[width=0.8\textwidth,clip = true, trim=1.0cm 0.5cm 1.0cm 0.5cm,keepaspectratio]{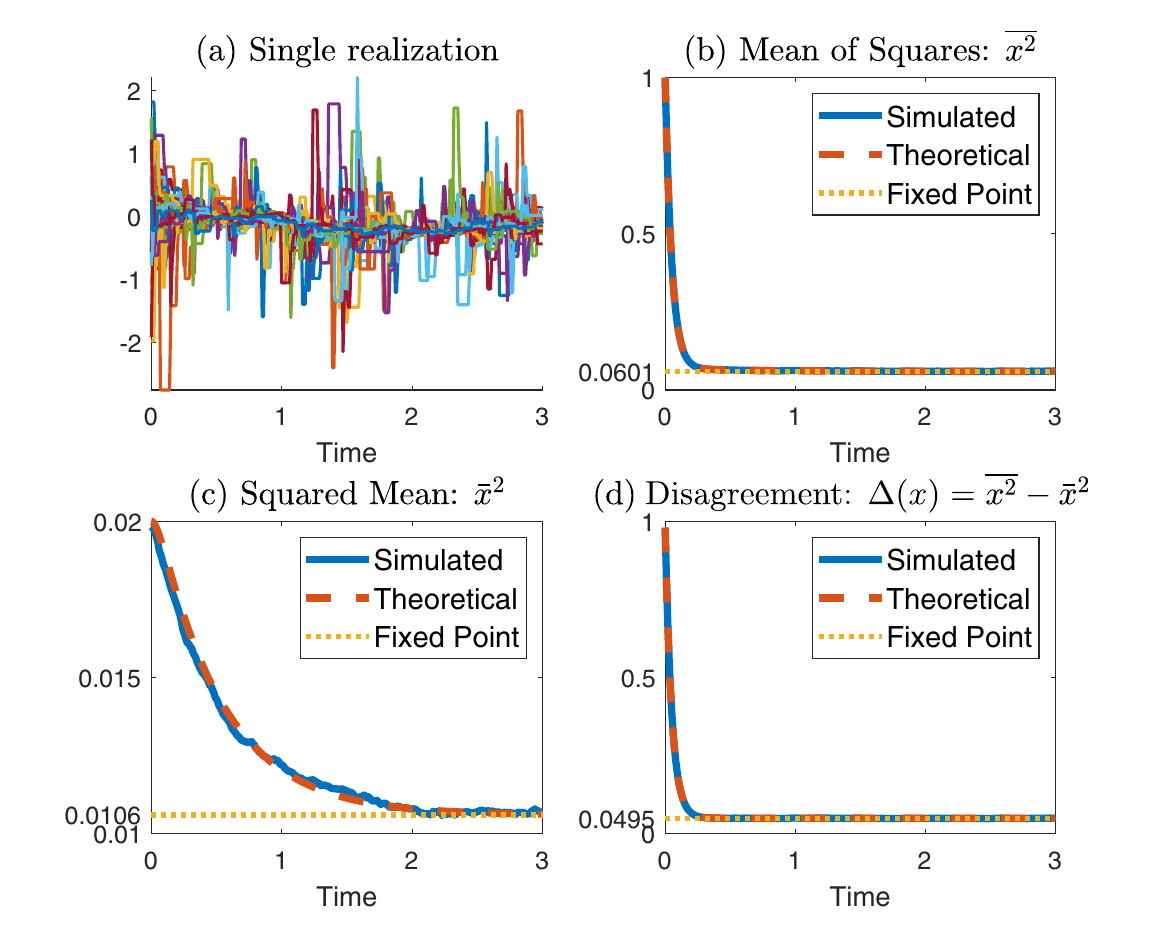}
    \caption{Behavior of an open system of $50$ agents subject to random replacements and gossips with rates $\lambda_r=1$ and $\lambda_g=19$, so that $\lambda_r+\lambda_g=20$. (a) depicts the evolution of one realization of the system, where each line is the estimate of an agent. (b), (c) and (d) respectively depict the evolution of the expected mean of squares $\MoS$, squared mean $\SqM$ and disagreement $\Var(x)$ simulated over $10000$ realizations in plain blue line. Those are compared with the theoretical results as obtained in Theorem~\ref{thm:FixedSize:RandomRepl:sys-evolution} in dashed red line (the disagreement is deduced from $\Var(x) = \MoS-\SqM$). The expected fixed points from equations \eqref{eq:FixedSize:RandomRepl:FixedPoint:SqM&MoS} and \eqref{eq:asymptotic_var_constant_n} are also depicted in yellow dotted line.}
    \label{fig:FixedSize:RandomRepl:Estimates&Descriptors:n50}
\end{figure}

\subsection{Varying-size systems}
\label{sec:App:VarSize}
We now consider a more general case where arrivals and departures are decoupled, so that the system size changes with time.
The rate of arrivals is independent of the system size, whereas the rates of gossip interactions and departures scale linearly with it, so that when the system size is $n$ one has $\lambda_{\An}=\lambda_a$, $\lambda_{\Gn}=n\lambda_g$ and $\lambda_{\Dn}=n\lambda_d$ for some constants $\lambda_a,\lambda_d,\lambda_g\geq0$ (and $\lambda_{\Rn}=0$ as we assume no replacement happens).
It implies that the rates of gossip and departures per individual agent remain independent of $n$, while the general arrival rate remains constant.
We could in principle derive an exact expression for the evolution of the descriptors following the same approach as for the fixed-size system.
However, for the sake of simplicity, we will directly work on the expected disagreement, only obtaining thus an upper bound on its evolution.

We also note that the next theorem will use Proposition \ref{prop:weak_deriv} on weak derivative of infinite sums of $V_j\pi_j$, for convenience we will use classical derivative notations.
\smallskip

\begin{theorem}
\label{thm:VarSize:sys-evolution}
    In a system of variable size subject to arrivals, departures and gossips as described above, there holds
    \begin{align}
        \label{eq:thm:VarSize:sys-evolution}
        \frac{d}{dt} \E \Var(x(t))
        \leq -\lambda_g \E \Var(x(t)) + \E\left[\frac{\lambda_a}{n(t)+1}\right]\sigma^2.
    \end{align}
\end{theorem}
\begin{proof}
    Let $f(S(t)) = \Var(x(t))$. 
    The result follows from applying 
   Corollary \ref{cor:evol_affine}  
   combined with Lemmas~\ref{lem:effect_operations:gossip} to \ref{lem:effect_operations:departure}, yielding for any $j\geq 1$
    \begin{align*}
        \frac{d}{dt}V_j\pi_j 
        &\leq -\prt{\lambda_a+\lambda_g+j\lambda_d}V_j\pi_j + \lambda_a\tfrac{j-1}{j}V_{j-1}\pi_{j-1}\\
        &\ \ \ + j\lambda_d\prt{1-\tfrac{1}{j^2}}V_{j+1}\pi_{j+1} + \lambda_a\pi_{j-1}\tfrac{\sigma^2}{j},
    \end{align*}
    with $V_j := \E\prt{\Var(x(t)) \mid n(t)=j}$ and $\pi_j=P\prt{n(t)=j}$, where we omit the dependence to the time to lighten the notations and remind that $V_j=0$ by convention.
    Summing up over all values of $j\geq 0$, noting $\frac{d}{dt}V_0\pi_0=0$ then yields
    \begin{align*}
        \frac{d}{dt}\E V
        &\leq -\lambda_g\E V + \sum_{j=0}^\infty \lambda_a\tfrac{\pi_j}{j+1}\sigma^2 - \sum_{j=2}^\infty \prt{\tfrac{\lambda_a}{j+1}+\tfrac{j\lambda_d}{(j-1)^2}}V_j\pi_j,
    \end{align*}
    where $V = \Var(x(t))$.
    Since the last term is nonpositive, it can be bounded by $0$, which concludes the proof.
\end{proof}

The result above provides an upper bound on the expected evolution of the disagreement and is parallel to that provided in \eqref{eq:FixedSize:RandomRepl:BoundOnVar} for the fixed-size system case.
Interestingly, it overestimates the effect of arrivals on the disagreement, while neglecting the slight favorable impact of departures, thus introducing conservatism in the bound.

Analyzing the asymptotic behavior of \eqref{eq:thm:VarSize:sys-evolution} requires characterizing the evolution of the system size $n(t)$. Observe that it
follows a birth-death process as depicted in Figure~\ref{fig:SystemSize:BirthDeathProcess} with birth and death rates respectively defined by the arrival and departure rates.
Hence, we have the following lemma, which also shows that $\nn = \lambda_a/\lambda_d$ is the characteristic system size.

\smallskip

\begin{lemma}
\label{lem:VarSize:BirthDeathProcess}
    Let $\nn:=\lambda_a/\lambda_d<\infty$, then for all $i\in\N$ there holds $\pi_i^* = \lim_{t\to\infty}\pi_i(t) = \frac{\nn^i}{i!}e^{-\nn}$. In particular, $\lim_{t\to\infty}\E n(t)=\nn$.
\end{lemma}
\begin{proof}
    Standard results on birth-death processes
    show that as $n(t)$ is ergodic, then $\pi_i^* = \lim_{t\to\infty}\pi_i(t)$ exists for all $i$ and satisfies $\pi_i^* = \pi_0^*\prod_{j=1}^i \frac{\lambda_a}{j\lambda_d}$ with
    \begin{align*}
        \pi_0^* 
        = \frac{1}{1+\sum_{k=1}^\infty \prod_{j=1}^k \frac{\lambda_a}{j\lambda_d}}
        = \frac{1}{1+\sum_{k=1}^\infty \frac{\nn^k}{k!}}
        = e^{-\nn},
    \end{align*}
    yielding $\pi_i^* = \frac{\nn^i}{i!}e^{-\nn}$.
    The conclusion then follows from the fact that $n(t)$ is ergodic if $\lambda_a/\lambda_d<\infty$.
\end{proof}

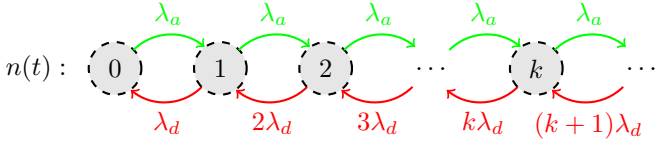
\begin{figure}
    \centering
    \tikzstyle{State}=[thick, dashed, fill=gray!20]
    \begin{tikzpicture}[scale=0.7]
        \draw (-7.5,0) node {$n(t):$};
        
        \draw [State] (-6,0) circle(0.5);
        \draw (-6,0) node {$0$} ;
        \draw[->, green, thick] (-6+0.35,0.35) arc(135:50:1) node[midway, above] {$\lambda_a$} ;
        \draw[->, red, thick] (-6+1.65,-0.35) arc(315:230:1) node[midway, below] {$\lambda_d$} ;
        
        \draw [State] (-4,0) circle(0.5);
        \draw (-4,0) node {$1$} ;
        \draw[->, green, thick] (-4+0.35,0.35) arc(135:50:1) node[midway, above] {$\lambda_a$} ;
        \draw[->, red, thick] (-4+1.65,-0.35) arc(315:230:1) node[midway, below] {$2\lambda_d$} ;
        
        \draw [State] (-2,0) circle(0.5);
        \draw (-2,0) node {$2$} ;
        \draw[->, green, thick] (-2+0.35,0.35)    arc(135:50:1) node[midway, above] {$\lambda_a$} ;
        \draw[->, red, thick] (-2+1.65,-0.35) arc(315:230:1) node[midway, below] {$3\lambda_d$} ;
        
        \draw (0,0) node {$\ldots$};
        \draw[->, green, thick] (0.35,0.35) arc(135:50:1) node[midway, above] {$\lambda_a$} ;
        \draw[->, red, thick] (1.65,-0.35) arc(315:230:1) node[midway, below] {$k\lambda_d$} ;
        
        \draw [State] (2,0) circle(0.5);
        \draw (2,0) node {$k$} ;
        \draw[->, green, thick] (2+0.35,0.35) arc(135:50:1) node[midway, above] {$\lambda_a$} ;
        \draw[->, red, thick] (2+1.65,-0.35) arc(315:230:1) node[midway, below] {$(k+1)\lambda_d$} ;
        
        \draw (4,0) node {$\ldots$};
    \end{tikzpicture}
    
    \caption{Continuous-time Markov Chain representing the evolution of the system size as a birth-death process where the birth and death rates are respectively given by $\lambda_a$ and $j\lambda_d$ when the system size is $n(t)=j$.}
    \label{fig:SystemSize:BirthDeathProcess}
\end{figure}

We can now bound the asymptotic expected disagreement.
\smallskip

\begin{corollary}
\label{cor:VarSize:AsymptoticEVar}
    Let $\nn=\lambda_a/\lambda_d$ and $\gamma = \lambda_g/\lambda_d$, there holds
    \begin{align}
        \label{eq:cor:VarSize:AsymptoticEVar}
        \lim_{t\to\infty}\E\Var(x(t)) \leq \frac{1-e^{-\nn}}{\gamma}\sigma^2.
    \end{align}
\end{corollary}
\begin{proof}
    We define $h(t) := \E\left[\frac{\lambda_a\sigma^2}{n(t)+1}\right] = \lambda_a\sum_{j=0}^\infty \frac{\pi_j(t)}{j+1}\sigma^2$.
    Then, using Lemma~\ref{lem:VarSize:BirthDeathProcess}, for $\lambda_d\neq0$ we have
    \begin{align*}
        \lim_{t\to\infty} h(t) = \lambda_ae^{-\nn}\sum_{j=0}^\infty \frac{\nn^j}{(j+1)!}\sigma^2 = \lambda_a\frac{1-e^{-\nn}}{\nn}\sigma^2.
    \end{align*}
    Gr\"onwall's lemma then yields
    \begin{align*}
        \E\Var(x(t))
        \leq e^{-\lambda_gt}\E\Var(x(0)) + \int_0^th(s)e^{-\lambda_g(t-s)}\mathrm ds.
    \end{align*}
    Standard results from dynamical systems then yield
    \begin{align*}
        \lim_{t\to\infty}\int_0^th(s)e^{-\lambda_g(t-s)}\mathrm ds 
        = \lim_{t\to\infty} \frac{h(t)}{\lambda_g} 
        = \frac{\lambda_a}{\lambda_g}\frac{1-e^{-\nn}}{\nn}\sigma^2,
    \end{align*}
    which concludes the proof.
\end{proof}

Figure \ref{fig:VarSize:AsymptoticEVar} compares the result of Corollary~\ref{cor:VarSize:AsymptoticEVar} with simulated results for a characteristic number of agents $\bar n=20$. One can see that our bounds accurately predicts the expected disagreement, as soon as $\gamma=\lambda_g/\lambda_d$ is not too small.
Interestingly, adding a factor 1 at the denominator as in \eqref{eq:asymptotic_var_constant_n} seems to lead to a much more accurate bound, but has no theoretical justification at this stage. Both our bound and this postulated one were observed to be less accurate when both $\gamma$ and the characteristic number of agents are very small $\bar n = 1,2$, but this regime is less interesting.

We identify two sources for the conservatism of our bound in some regimes. First, our analysis focused solely on the disagreement dynamics instead of the coupled dynamics of $\SqM$ and $\MoS$, leading to a less accurate (and more conservative) upper bound, as also observed in Remark \ref{rem:2dyn_vs_1dyn} in the constant-size case. Second, our analysis of the infinitely many coupled dynamics of the conditional disagreements $V_j$ was performed using a relatively simple technique, summing all the $\pi_jV_j$, which invovled neglecting some terms when obtaining the bound of Theorem \ref{thm:VarSize:sys-evolution}. More elaborated techniques of analysis could yield tighter result, see e.g. Proposition 4.18 in \cite{de2022open}
but are out of the scope of the current study.

\begin{figure}
\centering
        \includegraphics[scale =.45]{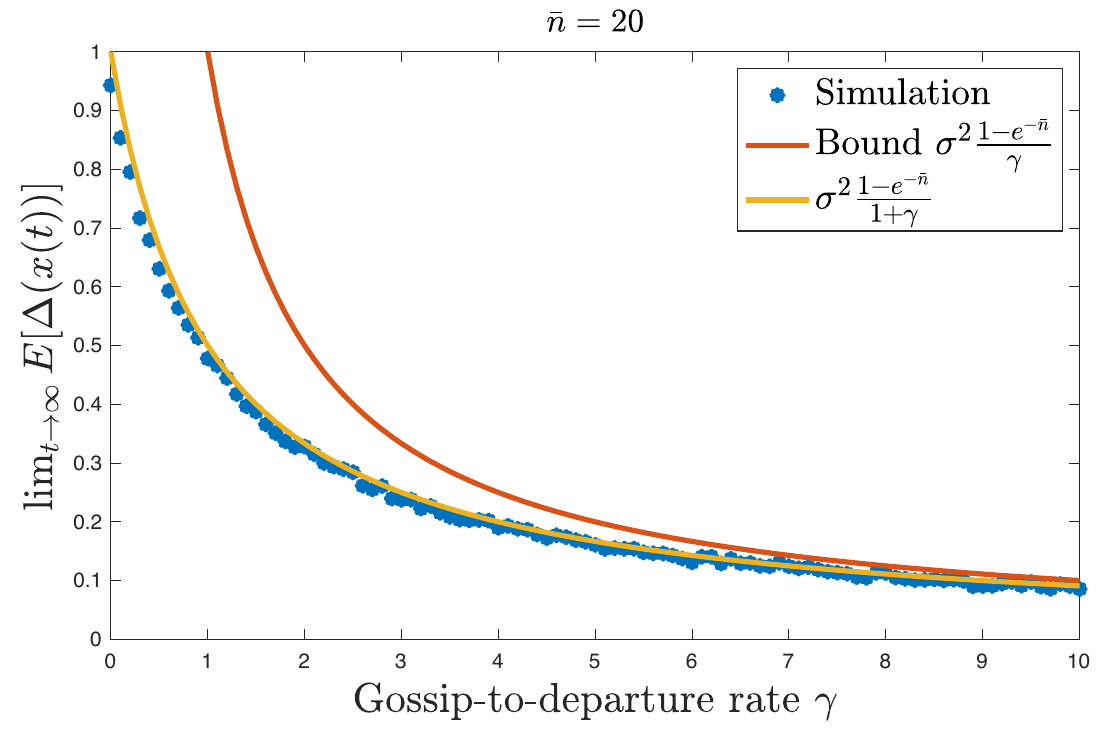}\\

    \caption{Evolution with $\gamma=\lambda_g/\lambda_d$. of the expected asymptotic disagreement in the setting of Section \ref{sec:App:VarSize}, for $\lambda_d=1$ and $\lambda_a=20$ and thus a characteristic number of 20 agents. 
    Each point for the simulation (in blue) corresponds to the average over 1000 trials, with the situation at $t=50$ being taken as \quotes{asymptotic value}. It is compared with the theoretical result \eqref{eq:cor:VarSize:AsymptoticEVar} from Corollary~\ref{cor:VarSize:AsymptoticEVar}, and a similar function that appears numerically more accurate.
    }
    \label{fig:VarSize:AsymptoticEVar}
\end{figure}

Coming back to the dependency of bound in Corollary~\ref{cor:VarSize:AsymptoticEVar} on the system parameters, we see that the influence of the characteristic number of agents rapidly vanishes as soon as $\nn$ is not very small. The bound on expected disagreement is then approximately $\frac{\sigma^2}{\gamma}$, i.e. inversely proportional to $\gamma=\lambda_g/\lambda_d$, the expected number of gossip interactions involving a given agent before it leaves the system, similarly to \eqref{eq:asymptotic_var_constant_n} obtained in Section~\ref{sec:App:FixedSize} for the fixed-size case (up to a constant term). The next corollary takes this observation to the limit $\lambda_d=0$, when agents never leave. It shows that in such systems, although the agents never actually reach consensus, the expected disagreement tends to zero as time passes.

\smallskip

\begin{corollary}
    In the same context as Theorem \ref{thm:VarSize:sys-evolution} but with no departure, $\lambda=0$, $\lim_{t\to \infty}\Var(x(t))=0$.
\end{corollary}
\begin{proof}
    When $\lambda_d=0$, the number of agents forming a (continuous-time) biased random walk, $n(t)$ will exceed any threshold with a probability that tends to one when time grows. 
and consequently $\lim_{t\to\infty } \E\prt{\frac{1}{1+n(t)}}=0$. The result follows then directly from Theorem \ref{thm:VarSize:sys-evolution} and Gronwall's inequality.
\end{proof}

\section{Conclusions}
\label{sec:Conclusion}

We considered open systems where agents can join and leave, which raise several new challenges including dimension variations and the absence of usual convergence.
We have shown that such systems subject to agent interactions can be characterized in terms of scale-independent descriptors whose evolution is governed by coupled linear systems whose dimension is independent of the number of agents.

We applied this approach to analyze the behavior of all-to-all pairwise gossips in fixed-size systems subject to replacements, 
and have shown that the expected evolution of the two first moments (and thus of a second-order variance-like measure of the disagreement) is characterized exactly. We observed that following two moments could give more accurate information about that disagreement than if we had just followed the disagreement directly. We further analyzed the evolution of the expected disagreement in systems whose size does change (stochastically) with time. 

Our approach relies on an extension of certain classical results on continuous-time Markov processes on $(n(t),x(t))$ to allow for infinite number of discrete states $n(t)$, unbounded dimensions for the variable $x(t)\in \R^{n(t)}$, and jump rates that are not uniformly bounded. These results may be of independent interest, and allow treating more complicated open multi-agent systems, or other systems with variable state-space dimension. We expect it might be possible to further relax some of our technical assumptions, such as the bounded support of the initial distribution of $n$, or part of the bounds on the rate within a distance 2 in Assumption \ref{assum:finite_sum_proba}.

We believe our methodology can also be applied to the analysis of interactions restricted to a graph: the challenge is then to find an appropriate set of descriptors that exploit the graph topology, and tight results might not be guaranteed with an inappropriate choice (\textit{e.g.}, for complete bipartite graph, an appropriate choice would be the two first moments for each parts as well as the expected product of the average of the two parts).
Other challenges left untackled consist in describing the actual probability distribution of our descriptors instead of only their expected value, and to incorporate the possibility of noise and delay in the analysis.

\appendix

\section{Upper bound on the descriptors}
\label{sec:appendix:bound_squaredmean}

\begin{proof}
    Fixing $t\geq 0$, let us use the shorthand notation $\NN=\NN(t)$ for the active agents at the target time.
    We will prove  
    \begin{equation}
    \begin{aligned}
    \E \prt{\prt{\sum_{i\in \NN} x_i(t)}^2\ \bigg|\ \NN}\leq  \abs{\NN}\sigma^2, \\
    \E \prt{\sum_{i\in \NN} x_i(t)^2\ \bigg|\ \NN} \leq  \abs{\NN}\sigma^2,
    \end{aligned}\label{eq:bound_sigma_no_scale}
 \end{equation}
which will directly allow expressing $\E(\SqM \mid |\NN|)$ as
    \[
    \E\prt{\E \prt{\prt{\dfrac{1}{\abs{\NN}}\sum_{i\in \NN} x_i(t)}^2\ \Bigg|\ \NN} \Bigg|\ |\NN|} \leq  \frac {\sigma^2}{\abs{\NN}}, 
    \]
    and an analogous reasoning for $\MoS$, relying on the bound holding for $\E(\cdot \mid \NN)$.
    
    Let $\T$ be the (random) set of the labels of all the agents that have been in the system at some time between times $0$ and $t$ (hence $\NN\subseteq\T$).
    Moreover, let $\alpha_i$ and $\delta_i$ respectively denote the arrival and departure times of agent $i$ in the system, and we define the vector $\xi(\tau)\in\R^{|\mathcal T|}$ for $0\leq \tau \leq t$ such that for all $i\in\mathcal T$:
    \begin{equation*}
        \xi_i(\tau) := 
        \begin{cases}
            x_i(\alpha_i) &\hbox{if } \tau\leq \alpha_i,\\
            x_i(\tau) &\hbox{if } \alpha_i \leq \tau \leq \delta_i,\\
            x_i(\delta_i) &\hbox{if } \delta_i\leq \tau.
        \end{cases}
    \end{equation*}
    Arrivals and departures have no direct effect on $\xi(\tau)$, whereas gossip interactions result in the multiplication of $\xi$ by a doubly-stochastic matrix.
    Let $A$ be the product of all those matrices since time $0$, then $A$ is doubly-stochastic, and there holds $\xi(t) = A\xi(0)$.
    One then has $\sum_{i\in \NN}x_i(t)=\sum_{i\in \NN,j\in\T} A_{ij}\xi_j(0)=\sum_{j\in\T}w_j\xi_j(0)$ with $w_j = \sum_{i\in S}A_{ij}$.

    Since $A$ is doubly-stochastic, and since $\NN\subseteq\T$, one has $w_j \in [0,1]$, and thus $w_j^2\leq w_j$. 
    Moreover, there holds $\sum_{i\in \NN}A_{ij}^2 \leq (\sum_{i\in \NN}A_{ij})^2 = w_j^2$, and hence
    \begin{align}
    \label{eq:proof:lem:tools:Bound:intermediateResult:SqM}
        &\E\prt{\prt{\sum_{i\in \NN} x_i(t)}^2 \bigg| \  \NN}
        = \E\prt{\prt{\sum_{j\in\T}w_j\xi_j(0)}^2 \bigg| \  \NN}\nonumber\\
        &\ \ = \sum_{j\in\T}\E\prt{w_j^2\xi_j(0)^2 \mid \NN} 
        \leq \sigma^2\sum_{j\in\T}\E(w_j\mid \NN),
        \end{align}
    where the absence of correlation between the initial values $\xi_j(0)$ was used to nullify the cross-products, and the independence of initial conditions from subsequent events to separate $\sigma^2$. Similarly,
    \begin{align}
        &\E\prt{\sum_{i\in \NN} x_i(t)^2 \ \bigg|\ \NN}
        = \E\prt{\sum_{i\in \NN}\prt{\sum_{j\in\T}A_{ij}\xi_j(0)}^2 \bigg| \  \NN}\nonumber\\
        \label{eq:proof:lem:tools:Bound:intermediateResult:MoS}
        &\ \ = \E\prt{\sum_{j\in\T,i\in \NN}A_{ij}^2\xi_j(0)^2 \  \bigg| \  \NN} \leq \sigma^2\sum_{j\in\T}\E(w_j\mid \NN).
    \end{align}
     Finally, since $\sum_{j\in\T}w_j = \sum_{i\in \NN}\sum_{j\in\T}A_{ij}=|\NN|$, \eqref{eq:proof:lem:tools:Bound:intermediateResult:SqM} and \eqref{eq:proof:lem:tools:Bound:intermediateResult:MoS} lead to  \eqref{eq:bound_sigma_no_scale}, which we have argued establishes the results jointly for all $|\NN|=n(t)=j$ cases.   
\end{proof}

\section{Effect of events}
\label{sec:appendix:EffectOfEvents}

\subsection{Proof of Lemma~\ref{lem:effect_operations:gossip} (Gossip step)}

    Let us first fix the nodes $i,j$ involved in the gossip step. 
    Firstly, observe that $x_i^++x_j^+ = 2\frac{x_i+x_j}{2}=x_i+x_j$, and that $x_k^+=x_k$ for all $k\neq i,j$. 
    Hence $\bar x^+ = \bar x$, which establishes the fist line of \eqref{eq:lem:effect_operations:gossip:X}. 
    Secondly, since $x_k=x_k^+$ for every $k\neq i,j$, there holds
    \begin{align}\label{eq:gossip_avg_x2}
        \overline{(x^+)^2} 
        &= \frac{1}{n}\sum_{k=1}^n(x_k^+)^2 
        = \overline{x^2} + \frac{1}{n}\prt{ 2\prt{\tfrac{x_i+x_j}{2}}^2- x_i^2-x_j^2} \nonumber \\
        &= \overline{x^2} + \tfrac{1}{n}\prt{x_ix_j - \tfrac{1}{2}x_i^2-\tfrac{1}{2}x_j^2}.
    \end{align}
    Observe that $\E(x_i^2|x)= \E(x_j^2| x)= \overline{x^2}$ and $\E(x_ix_j| x)=\bar x^2$. 
    Taking the expectation with respect to $i$ and $j$ in \eqref{eq:gossip_avg_x2} yields
    \begin{align*}
        \E\prt{\overline{(x^+)^2}\mid x}
        = \prt{1-\tfrac{1}{n}}\overline{x^2} + \tfrac{1}{n} \bar x^2, 
    \end{align*}
    from which the second line of \eqref{eq:lem:effect_operations:gossip:X} follows. 
    The result on the disagreement $\Var{ }$ follows from computing $\begin{pmatrix}-1&1\end{pmatrix}\E\prt{X^+ \mid X,\Gn}$.
    
\subsection{Proof of Lemma~\ref{lem:effect_operations:arrival} (Arrival)}
    We suppose first $n>0$ and label $n+1$ the arriving agent for simplicity, so that $x_k^+=x_k$ for all $k\leq n$.  
    We begin by computing the new average :
    \begin{align}\label{eq:arrival_average}
        \bar x^+
        &= \frac{\sum_{k=1}^{n}x_k+ x^+_{n+1}}{n+1}
        = \frac{n}{n+1} \bar x + \frac{1}{n+1}x^+_{n+1}.
    \end{align}
    Since $\E x^+_{n+1} = 0$, we have $\E(\bar x^+\mid x) = \frac{n}{n+1}\bar x$. 
    By exactly the same reasoning but using $\E(x^+_{n+1})^2 = \sigma^2$  we also obtain
    \begin{align}\label{eq:arrival_x^2}
        \E\prt{\overline{(x^+)^2} \mid x} 
        &= \frac{n}{n+1}\overline{x^2} + \frac{1}{n+1}\sigma^2,
    \end{align}
    from which the second line of \eqref{eq:lem:effect_operations:arrival:X} follows.
    Turning to the first line, we obtain from \eqref{eq:arrival_average}
    \begin{align*}
        \E((\bar x^+)^2 \mid x) 
        =& \frac{1}{(n+1)^2}\prt{n^2\SqM + n \bar x \E x^+_{n+1}+\E (x_{n+1}^+)^2} \\
        =& \frac{n^2}{(n+1)^2} \SqM + 0  +\frac{1}{(n+1)^2} \sigma^2.
    \end{align*}
    Computing $\begin{pmatrix}-1&1\end{pmatrix}\E\prt{X^+ \mid X,\Gn}$ then yields the disagreement, where we use Lemma~\ref{lem:sys-description:bound_on_expected_quantities} to obtain the inequality. Finally, in case $n=0$, the result also holds because $X^+ = \begin{pmatrix}\sigma^2&\sigma^2\end{pmatrix}^\top$ and $X=\begin{pmatrix}0&0\end{pmatrix}^\top$.

\subsection{Proof of Lemma~\ref{lem:effect_operations:departure} (Departure)}  
    If $n=1$ the result holds by the definition as $X=0$ for $n=0$. 
    Suppose thus $n>1$ and let $j$ be the randomly selected agent that leaves the system. 
    It follows
    \begin{align}\label{eq:mean_departure}
        \bar x^+ = \frac{1}{n-1}\prt{\prt{\sum_{k=1}^n x_k} - x_j}
        = \frac{1}{n-1}\prt{n\bar x - x_j}.
    \end{align}
    By exactly the same reasoning, there holds
    $\overline {(x^+)^2}= \frac{1}{n-1}\prt{n \overline{x^2} - x_j^2}$.
    Since $j$ is randomly selected, $\E(x_j^2|x) = \overline{x^2}$. 
    Hence, 
    $\E \prt{\overline {(x^+)^2} \mid x} = \frac{1}{n-1}\prt{n \E \overline{x^2} - \E \overline{x^2}}=\E \overline {x^2}$,
    which implies the second line of \eqref{eq:lem:effect_operations:departure:X}. 
    For the first line, taking into account $\E(x_j \mid x) = \bar x$, it follows from \eqref{eq:mean_departure} that
    \begin{align*}
        \E ( (\bar x^+)^2\mid x)
        &= \tfrac{1}{(n-1)^2}\prt{n^2\SqM - 2n\bar x\E(x_j\mid x) + \E(x_j^2\mid x)}\\
        &=  \tfrac{n^2-2n}{(n-1)^2}\SqM  + \tfrac{1}{(n-1)^2} \MoS.
    \end{align*}
    Computing $\begin{pmatrix}-1&1\end{pmatrix}\E\prt{X^+\mid X,\Gn}$ yields the disagreement.
    
\subsection{Proof of Lemma~\ref{lem:effect_operations:replacement} (Replacement)}
    The matrix equality \eqref{eq:lem:effect_operations:replacement:X} follows from a combination of Lemmas \ref{lem:effect_operations:departure} and \ref{lem:effect_operations:arrival}, the latter applied to a system of size $n-1$ joined by an $n^{th}$ agent.
    The inequality on the disagreement follows from
    \begin{align*}
        \E(\Var(x^+)\mid & ~x,\Rn)\\
        &= -\frac{n-2}{n} \bar{x}^2 - \frac{1}{n^2}\overline{x^2} - \frac{\sigma^2}{n^2} 
        +\frac{n-1}{n} \overline{x^2} + \frac{\sigma^2}{n}\\
        &= \frac{n-1}{n^2} (\bar{x}^2+\sigma^2) 
         +\frac{n^2-n-1}{n^2} \Var(x)\\
         &\le
         \frac{n^2-n-1}{n^2} \Var(x)
         + \frac{(n^2-1)\sigma^2}{n^3},
    \end{align*}
    where we used Lemma~\ref{lem:sys-description:bound_on_expected_quantities} for the last inequality.

\section{Proofs of the ODE descriptions}\label{sec:proof_ode}

\subsection{Boundedness}

We first collect some boundedness properties that support the analysis to come.
For this purpose, we call an event \emph{effective} if it is such that the Poisson clock $N_\epsilon$ ticks at a time $t$ when $n(t)$ was $\source{\epsilon}$, and thus it indeed effects the state (as opposed to when it ticks and nothing happens because $n(t)$ is not the right one). When clear from the context we will omit the word \quotes{effective}.

Furthermore we need to handle multiple events occurring, let us thus introduce a new object; an \emph{event order} $\evorder$ is a sequence of event types $(\epsilon_1,\epsilon_2,\dots)$. 

We say that an event order $\evorder$ \emph{effectively occurred} between time $t$ and $t+\tau$ if the sequence of effective events that occurred after time $t$ and until $t+\tau$ 
corresponds exactly to the sequence prescribed by $\evorder$. We will denote this by $\evorder^{(t,t+\tau]}$, with the asymmetry of the interval is selected to be consistent with the right-continuity of the process.

An event order effectively occurring implies in particular that $n(t) = \source{\epsilon_1}$, $\arrival{\epsilon_1} = \source{\epsilon_2}$, etc.

We have then a useful immediate consequence of the bound on $\|A_\epsilon\|$ by Assumption \ref{assumption:bounded_matrix} and of \eqref{eq:prop:tools:DescriptorsEvolution:ode:Kolmogorov:LinearSys}.

\begin{lemma}\label{lem:trivial_implication_boundedmatrices}
    Under Assumption \ref{assumption:bounded_matrix} consider any state $S\in E$ and let $\evorder$ be any event order.
    Then 
 \begin{equation}
||\E (f(S(t+\tau)) \mid \evorder^{(t,t+\tau]}, S(t)=S)|| \leq || f(S)||.
\end{equation}
Moreover, this relation also holds if \eqref{eq:prop:tools:DescriptorsEvolution:ode:Kolmogorov:LinearSys} is relaxed to a coordinate-wise inequality, provided the norm is monotonous. 

\end{lemma}

The next technical result shows that the initial bound in expectation by Assumption \ref{assumption:bounded_initial} can be carried through in time.

\smallskip

\begin{lemma}\label{lem:bound_barF}
    Under Assumption \ref{assumption:bounded_matrix} and \ref{assumption:bounded_initial}, for every $k$ and $t$ there holds
  \begin{align*}
    \norm{F_k(t)} = \norm{\E (f(S(t)) \mid n(t)=k)  }  &\leq \bar F,\\
    \E( \norm{f(S(t))} \mid n(t)=k)\leq K\norm{F_k(t)}& \leq  K\bar F.
    \end{align*}
\end{lemma}
\begin{proof} 
The first equality follows directly from the definition $F_k(S(t)):=\E (f(S(t)) \mid n(t)=k)  $.
Let us now assume that we start (deterministically) at $n(0)=j$, and let $\evorderset_{j\to k}$ be the collection of event orders $\evorder$
that can effectively occur and would drive $n$ from $j$ to $k$. We have
\begin{align*}
||\E &\prt{ f(S(t))  \mid n(t) = k}||\\
&= 
||\E_{\evorder \in \evorderset_{j,k},S\in E_j}  \prt{\E \prt{ f(S(t))  
\mid \evorder^{(0,t]}, S(0)=S}}||\\
&\leq 
\E_{\evorder \in \evorderset_{j,k},S\in E_j} \prt{ ||\E \prt{ f(S(t))  
\mid \evorder^{(0,t]}, S(0)=S}||}\\
&\leq \E_{\evorder \in \evorderset_{j,k},S\in E_j} || f(S(0))|| \\
&= \E_{S\in E_j} || f(S(0))|| \\
&= \E(||f(S(0))||) \leq \bar F,
\end{align*}
reminding that we assumed $n(0)=j$, and using Lemma \ref{lem:trivial_implication_boundedmatrices} for the second inequality and Assumption \ref{assumption:bounded_initial} for the last one.
Here we gradually take expectations in variables appearing, w.r.t.\ their distribution already determined.

Furthermore, if $\pi(0)$ is not concentrated on a single value $j$, a mixture of the above, with a uniform upper bound in $j$, confirms the primary claim.  Finally, note by convexity and the equivalence of norms (of finite dimension) the expected norm can also be bounded above at an extra factor $K$. 
Indeed, using the $L^1$ norm $\|\cdot\|_1$ and relying on the functions being non-negative, we can write
\begin{align*}
    \E( \norm{f(S(t))} &\mid n(t)=k)\\
    &\leq \E( K_1\norm{f(S(t))}_1 \mid n(t)=k)\\
    &=    \E( K_1f(S(t)) \mid n(t)=k)\\
    &=    K_1\norm{\E(f(S(t)) \mid n(t)=k)}_1\\
    &\leq K_1K_1'\norm{F_k(t)} \leq  K\bar F,
\end{align*}
with appropriate constants $K_1,K_1'$ and $K=K_1 K_1'$.
\end{proof}

\subsection{Differential of individual terms}

Let $f_j(s) := f(s) \1_{n=j}$, where we remind that $s$ denotes a state $(n,x)$. We have thus $\E(f_j(S(t))) = F_j(t)\P(n(t)=j).$ For simplicity of exposition, we further let 
$$ \Delta_j(s,\tau) := \E (f_j(S(t+\tau)) \mid S(t)=s) - f_j(s).$$
and $\Delta(s,\tau \mid X)$ the same but with the expected value also conditioned to any collection of events $X$. By time homogeneity the above definition does not depend on $t$.

We will show that  $\Delta_j(s,\tau)$ can be approximated by a linear function with a remainder that can be \emph{bounded independently of $s$ except via $f(s)$}, which is needed to bound the expected remainder of the linear approximation of $\Delta_j(s,\tau)$. 
For this purpose let $\E(f_j(s^+ \mid s))$ be the expectation of the value after the first effective event, starting from $s$ (note time-homogeneity), and let \begin{equation}\label{eq:bar_delta}
    \bar \Delta_j(s):=\E(f_j(s^+ \mid s)) - f_j(s).
\end{equation}
If no event occurs, i.e.\ $s$ is a sink, then we set $\E(f_j(s^+ \mid s)) = f_j(s)$ and $\bar \Delta_j(s) =0$.

We first need an intermediate lemma showing that the $\Delta$ are in many cases zero. For this purpose, let us
introduce the notation $n\sim n'$ if there exists an event $\epsilon$ with $s_\epsilon=n,a_\epsilon=n'$. Similarly, $n\not\sim n'$ if no such event exists. Also, let $\ev(s,t,\tau)$ denote the number of effective events occurred starting at state $s$ and time $t$ until $t+\tau$. For convenience we will often drop the parameters when clear from the context. The definition of of $f_j(s)$ then directly implies the following equalities. 

\smallskip

\begin{lemma}\label{lem:variation_0}
If $n(t)$ and $n(t+\tau)$ are both different from $j$, then 
$f_j(S(t+\tau)) = f_j(S(t))=0$, 
and therefore
\begin{align*}
    \Delta_j(s,\tau\mid \ev=0) &= 0\\
    \Delta_j(s,\tau \mid n(t),n(t+\tau)\neq j) &= 0\\
    \bar \Delta_j(s)&=0 \quad \forall n\not \sim j.
\end{align*}
\end{lemma}

We can now prove our main lemma on approximating the evolution of $f_j$.

\smallskip

\begin{lemma}
\label{lem:linear_approx_evol}
Let $f_j(s) = f(s) \1_{n=j}$ as before.
Then for every $j$ and every $0\leq\tau$, under Assumption \ref{assumption:bounded_matrix}, there holds
$$\Delta_j(s,\tau)=
\tau \lambda_{n(t),out} \bar \Delta_j(s) + R_j(s,\tau)  $$
with
$
\norm {R_j(s,\tau) } \leq K_j \norm{f(s)} \tau^2
$ for some constant $K_j$.
\end{lemma}

\begin{proof}
Let us decompose the left hand side, $\Delta_j(s,\tau)$ depending on the number of effective events happened.
\begin{align*}
\Delta_j(s,\tau) &= 0\cdot \P(\ev=0)\\
&~~+ \bar\Delta_j(s)P(\ev = 1)\\
&~~+ \Delta_j(s,\tau \mid \ev\geq 2)\P(\ev \geq 2)\\
&= \bar\Delta_j(s)P(\ev \geq 1)\\
&~~+ (\Delta_j(s,\tau \mid \ev\geq 2)-\bar \Delta(j))\P(\ev \geq 2).
\end{align*}
Comparing with the claim, we get the error term
\begin{align}
\!\!\!R_j(s,\tau) & \!=\! \bar \Delta_j(s)(\P(\ev \geq 1)-\tau \lambda_{n(t),out}) \label{eq:error_proba}\\
&\!+\!(\Delta_j(s,\tau \mid \ev\geq 2)-\bar \Delta_j(s))\P(\ev \geq 2) .\label{eq:error_2jumps}
\end{align}
To bound this expression, let us first consider \eqref{eq:error_proba}.

Observe that $\norm{\bar \Delta_j(s)}\leq 2 \norm{f(s)}$ due to Lemma \ref{lem:trivial_implication_boundedmatrices} and the triangle inequality.
Also $|P(\ev \geq 1)-\tau \lambda_{n(t),out}| \leq \tau^2\lambda^2_{n,out}$. Combining the two bounds would lead to a quadratic bound on \eqref{eq:error_proba}, but it would be dependent on $n$ and not only on $j$ (as the $\lambda_{n,out}$ are not uniformly bounded). However, Lemma \ref{lem:variation_0} implies that $\bar \Delta_j(s)$ is actually zero whenever $n\nsim j$, so we can bound \eqref{eq:error_proba} by 
$$
\tau^2 (\max_{n\sim j} \lambda_{n,out}^2) 2 \norm{f(s)},
$$
with now only depends on $j$.

We now move to \eqref{eq:error_2jumps}. First of all, the expression is 0 by Lemma \ref{lem:variation_0} if none of $n(t),n(t+\tau)$ or $n^+$ -- the size after the first jump -- are $j$. If this is not the case, $\ev\geq 2$ means that at least two events occurred (including non-effective ones) within distance 2 of $j$. The probability of this is bounded above by $\tau^2 \lambda_{j2}^2$. 

For any relevant $\evorder$, i.e.\ with at least two events within distance 2 of $j$, using again Lemma \ref{lem:trivial_implication_boundedmatrices} and the triangle inequality, we obtain
\begin{align*}
\|\Delta_j(s,\tau \mid \evorder^{[t,t+\tau]})\| \leq 2\|f(s)\|,\\
\|\bar \Delta_j(s)\| \leq 2\|f(s)\|.  
\end{align*}
Combining the factors above we arrive at the bound
 $$
 4\|f(s)\|\tau^2 \lambda_{j2}^2,
 $$
Including the other part of the error term above this confirms the claim with $K_j = 2\max_{n\sim j}\lambda^2_{n,out}+4\lambda^2_{j2}$, which is finite thanks to Assumption \ref{assum:finite_sum_proba}.
\end{proof}

Lemma \ref{lem:linear_approx_evol} shows that $\lambda_{n(t),out} \bar \Delta_j(S(t))$ is the relevant quantity for the time-evolution, and we now express its expectation.

\smallskip

\begin{proposition}
    \label{prop:deriv_value}
Under the assumptions of Theorem \ref{thm:evol_linear}, 
\begin{align} 
         \E ( \lambda_{n(t),out} \bar \Delta_j(S(t)))
        =  &\sum_{\epsilon:\arrival{\epsilon}=j}\lambda_\epsilon\pi_{\source\epsilon}(t)A_\epsilon F_{\source\epsilon}(t)\nonumber\\
        &- \sum_{\epsilon:\source{\epsilon}=j}\lambda_\epsilon \pi_j(t)F_j(t).
    \end{align}
if \eqref{eq:prop:tools:DescriptorsEvolution:ode:Kolmogorov:LinearSys} holds with equality, and the same relation holds with component-wise inequality if \eqref{eq:prop:tools:DescriptorsEvolution:ode:Kolmogorov:LinearSys} holds with component-wise inequality.
\end{proposition}
\begin{proof}
We first suppose \eqref{eq:prop:tools:DescriptorsEvolution:ode:Kolmogorov:LinearSys} holds with equality.
By definition of $\bar \Delta$  we have $\E \lambda_{n(t),out} \bar \Delta_j(s)=\E \lambda_{n(t),out} \E(f_j(s^+|s)) - \E \lambda_{n(t),out}f_j(s)$, with $s^+$ the state after the first event when starting the process from $s$. We will treat the two terms separately, starting with the second one.

Since $f_j(s) = f(s) \1_{n=j}$, we have 
\begin{align}
\E \lambda_{n(t),out} f_j(S(t))&= \E   \lambda_{n(t),out} f(S(t)) \1_{n(t)=j}\nonumber \\
&= \lambda_j \E(f(S(t)|n(t)=j) \pi_j(t)\\& = \lambda_j F_j(t)\pi_j(t) \label{eq:E_barDelta_2nd}
    \end{align}
by definition of $F_j$. We now move to the first term, conditioning first on $n(t)=n$. In that case, since $\lambda_{n,out} = \sum_{\epsilon:\source{\epsilon}=n} \lambda_\epsilon$, 
the probability of the first event being $\epsilon $ is $ \lambda_\epsilon/\lambda_{n,out}$ if $\source{\epsilon}=n(t)$, and 0 else.
Hence $\E \prt{\lambda_{n(t),out} \E(f_j(s^+|s))|n(t)=n}$ is
\begin{align}
&\E\prt{ \lambda_{n,out}  
\E(f(s^+) \1_{n^+=j} |s)
\big|n(t)=n}\nonumber\\
&= \E\left(\lambda_{n,out} \sum_{\epsilon, \source{\epsilon}=n,\arrival{\epsilon=j}} \frac{\lambda_\epsilon}{\lambda_{n,out} } \E (f(s^+)|s,\epsilon) \bigg| n(t)=n\right)\nonumber\\
&=  \E\prt{\sum_{\epsilon, \source{\epsilon}=n,\arrival{\epsilon=j}}  \lambda_\epsilon A_\epsilon f(s) \bigg| n(t)=n}\label{eq:making_prop_ineq}\\
&= \sum_{\epsilon, \source{\epsilon}=n,\arrival{\epsilon=j}} \lambda_\epsilon A_\epsilon \E (f(s) | n(t)=n) \nonumber\\
&= \sum_{\epsilon, \source{\epsilon}=n,\arrival{\epsilon=j}} \lambda_\epsilon A_\epsilon F_n(t),\nonumber
\end{align}
relying on the expected update being described by the linear maps $A_\epsilon$ as set in \eqref{eq:prop:tools:DescriptorsEvolution:ode:Kolmogorov:LinearSys}.
Taking now the expectation with respect to $n(t)$, we have, 
\begin{align*}
&    \E\prt{\E \prt{\lambda_{n(t),out} \E(f_j(s^+|s))|n(t)=n}}\\
&=\sum_n \pi_n(t) 
\sum_{\epsilon, \source{\epsilon}=n,\arrival{\epsilon=j}} \lambda_\epsilon A_\epsilon F_n(t)\\
&= \sum_{\epsilon, \arrival{\epsilon=j}} A_\epsilon F_{\source{\epsilon}}(t) \pi_\source{\epsilon}(t).
\end{align*}
Together with \eqref{eq:E_barDelta_2nd} they are both finitely summable thanks to Assumption \ref{assum:finite_sum_proba} and Lemma \ref{lem:bound_barF}, thus they can be combined to yield the desired result.

In case \eqref{eq:prop:tools:DescriptorsEvolution:ode:Kolmogorov:LinearSys} holds with component-wise inequality, exactly the same reasoning applies, replacing \eqref{eq:making_prop_ineq}, the only place where \eqref{eq:prop:tools:DescriptorsEvolution:ode:Kolmogorov:LinearSys} is used, with an 
inequality.
\end{proof}

There now only remains to properly link the differential with the expression found before.

\smallskip

\begin{proposition}
\label{prop:derivative_in_barDelta}
   Under the assumptions of Theorem \ref{thm:evol_linear}, we have
$$
\frac{d}{dt} \E (f_j(S(t)) =   \E ( \lambda_{n(t),out} \bar \Delta_j(S(t))), 
$$
with $\bar \Delta_j$ defined in \eqref{eq:bar_delta}.
\end{proposition}
\begin{proof}
Applying Lemma \eqref{eq:prop:tools:DescriptorsEvolution:ode:Kolmogorov:LinearSys} with $s= S(t)$ and taking expectation, we obtain
\begin{align}\label{eq:linear_approx_expect}
\E(f_j(S(t&+\tau))-\E (f_j(S(t))\\
&=\tau \E  ( \lambda_{n(t),out} \bar\Delta_j(S(t))) + \E R_j(S(t),\tau).
\nonumber\end{align}
This is all well-defined, and the linearity of expectation can be used: $f_j$ is non-negative, $\E (f_j(S(t)),\E (f_j(S(t+\tau))$ are finite by Lemma \ref{lem:bound_barF}, the term with $\bar \Delta_j$ is well-defined based on Proposition \ref{prop:deriv_value}, the final one on $R_j$ to be controlled now.
The last term can be bounded as
$$
\norm {\E R_j(S(t),\tau) } \leq K_j  \tau^2 \E \norm{f(S(t))}\leq K_j \tau^2 K\bar F,
$$
where the inequalities follow from Lemma \ref{lem:linear_approx_evol} and Lemma \ref{lem:bound_barF}.
Hence dividing \eqref{eq:linear_approx_expect} by $\tau$ and taking the limit yields the result.
\end{proof}

Theorem \ref{thm:evol_linear} follows then directly from Proposition \ref{prop:deriv_value} and Proposition \ref{prop:derivative_in_barDelta} once translating the notation as $\frac{d}{dt}(F_j(t)\pi_j(t)) = \frac{d}{dt} \E (f_j(S(t))$.

\section{Weak Overall Differentiability}\label{sec:weak_dif}

\begin{proof}[Proof of Proposition \ref{prop:weak_deriv}]
We will prove this Proposition in the context of Theorem \ref{thm:evol_linear}, i.e. for systems satisfying the linear relation \eqref{eq:prop:tools:DescriptorsEvolution:ode:Kolmogorov:LinearSys}. Generalization to the context of Corollary \ref{cor:evol_affine} is direct, as systems satisfying the affine relation \eqref{eq:prop:tools:DescriptorsEvolution:ode:Kolmogorov:AffineSys} can be recast into those satisfying \eqref{eq:prop:tools:DescriptorsEvolution:ode:Kolmogorov:LinearSys}, see the proof of Corollary \ref{cor:evol_affine}.

Since Lemma \ref{lem:bound_barF} establishes the $F_j$ are all uniformly bounded, and since the $w_j$ are bounded by hypothesis and the $\pi_j$ sum to 1, absolute convergence of $\sum_j w_j F_j \pi_j$ follows. Therefore we have
\begin{align*}
\prt{\sum_j w_j F_j(t+\tau)\pi_j(t+\tau) } - \prt{\sum_j w_j F_j(t)\pi_j(t) } \\= 
\sum_j w_j (F_j(t+\tau)\pi_j(t+\tau)- F_j(t)\pi_j(t)).
    \end{align*}
Now each $F_j \pi_j$ is differentiable thanks to Theorem \ref{thm:evol_linear}, so we have for every $j$
$$
(F_j\pi_j)(t+\tau) - (F_j\pi_j)(t) = \int_{t}^{t+\tau} (F_j\pi_j)'(u) du,
$$
Therefore we can combine as
\begin{align}\label{eq:before_fubini}
    \prt{\sum_j w_j (F_j\pi_j)(t+\tau) } - \prt{\sum_j w_j (F_j\pi_j)(t) }  \\ = \sum_j w_j\int_{t}^{t+\tau} (F_j\pi_j)'(u) du. \nonumber
\end{align}
Let us now show that $||w_j(F_j\pi_j)'(u)||$ is integrable over $[t,t+\tau]\times \mathbb{N}$, i.e.
$
\int_{t}^{t+\tau} \sum_j |w_j|\cdot || (F_j\pi_j)'(u)  || du < \infty,
$
or
\begin{equation}\label{eq:summability}
\int_{t}^{t+\tau} \sum_j || (F_j\pi_j)'(u)  || du < \infty,    
\end{equation}
since the $w_j$ are bounded. 
Following Theorem \ref{thm:evol_linear} we have the bound
\begin{align*}
    &\left\| ( F_j\pi_j)'(u) \right\|\\ 
    &\leq \sum_{\epsilon :\arrival{\epsilon}=j} \lambda_\epsilon \pi_{\source{\epsilon}}(u) ||A_\epsilon F_{\source{\epsilon}}(u)|| + \sum_{\epsilon:\source{\epsilon}= j} \lambda_\epsilon\pi_j(u) ||F_{j}(u)||\\
       &\leq \sum_{\epsilon :\arrival{\epsilon}=j} \lambda_\epsilon \pi_{\source{\epsilon}}(u) ||F_{\source{\epsilon}}(u)|| + 
    \sum_{\epsilon:\source{\epsilon}= j} \lambda_\epsilon\pi_j(u) ||F_{j}(u)||,
    \end{align*}
where we have used Assumption \ref{assumption:bounded_matrix}. It then follows from Lemma \ref{lem:bound_barF} that 
\begin{align*}
     \left\| ( F_j\pi_j)'(u) \right\|
    &\leq \sum_{\epsilon :\arrival{\epsilon}=j} \lambda_\epsilon \pi_{\source{\epsilon}}(u) \bar F  + 
    \sum_{\epsilon:\source{\epsilon}= j} \lambda_\epsilon\pi_j(u) \bar F \\
    & = \bar F\left[\sum_{\epsilon :\arrival{\epsilon}=j} \lambda_\epsilon \pi_{\source{\epsilon}}(u) + 
    \sum_{\epsilon:\source{\epsilon}= j} \lambda_\epsilon\pi_\source{\epsilon}(u)\right] .
\end{align*}
Adding this bound over all $j$ yields
\begin{align*}
    \sum_j    \left\| (F_j\pi_j)'(u)\right\| &\leq 2\bar F \sum_j  \sum_{\epsilon:\source{\epsilon}= j} \lambda_\epsilon\pi_j(u) \leq 2\bar F W(u),
\end{align*}
where the last bound follows from Assumption \ref{assumption:bounded_expected_rate}. As $W(u)$ is integrable on $[t,t+
\tau]$, \eqref{eq:summability} follows. Hence we can apply Fubini's Theorem to switch integral and summation in \eqref{eq:before_fubini}, leading to the desired result. 
\end{proof}

\section{Non-explosive behavior} \label{sec:non-explosive}

\begin{proof}[Proof of Lemma \ref{lem:expectedrate}]
    We show the claim when the initial distribution is deterministic $i$, then it immediately holds for a finite mixture. 

    By the definition of $h$ we can bound the target quantity in \eqref{eq:assumption_boundedexpectedrate} as $\sum_{\epsilon} \pi_{\source{\epsilon}}(t)\lambda_\epsilon\leq 
    \sum_j\pi_j(t)h_j=\E(h_{n(t)}) $.   

    We define an alternate process $\tilde n(t)$ with the same initial distribution but where decreasing events are disregarded, so that $n(t)$ can only increase and does this at a rate $\mu_{n(t)}$.
    We couple this process with $n(t)$ in the following way: when equal, use the same stream of events, when different, move independently. They can get different if $n(t)$ is decreased but $\tilde n(t)$ is not, but increase together when equal. Thus, for all $u\in\N$ and $t$
    \[
    \P(n(t)>u)\leq \P(\tilde n(t)>u).
    \]
    Since $h_u$ is nondecreasing, each $(h_{u+1}-h_u)$ is nonnegative. Hence we can bound the expectation of $h_{n(t)}$ as    
   \begin{align*}
        \E(h_{n(t)})&=\sum_{u=0}^\infty (h_{u+1}-h_u)\P(n(t)>u)\\
        &\leq \sum_{u=0}^\infty (h_{u+1}-h_u)\P(\tilde n(t)>u)
        = \E(h_{\tilde n(t)}),
    \end{align*}
    which it is therefore sufficient to bound. 
    This can be done using the result of \cite[Theorem 2.1]{meyn1993stabilitydynkin}. In that context, the evolving quantity of interest is the bound on the total rate, $V(j) = h_j$. The pure birth process $\tilde n(t)$ leads to the simple generator
    \[
    \mathcal A V(j) = (h_{j+1}-h_j)\mu_j.
    \]
    The theorem can be applied if $\mathcal A V(j) \leq C V(j)$ holds for all $j$ for some appropriate constant $C$, which in our case translates precisely to our condition\footnote{To be precise, $\mathcal A_m V(j) = (h_{j+1}-h_j)\mu_j$ holds for the truncation of the pure birth process to the first $m$ states (and $j<m$), for every $m$, and the condition of the theorem requires $\mathcal A_m V(j) \leq C V(j)$ for every such truncation.}
    \[
    (h_{j+1}-h_j)\mu_j\leq C h_j.
    \]
    Then \cite[Theorem 2.1]{meyn1993stabilitydynkin} shows that there is no explosion almost surely and the expected rate can be bounded as
    \[
    \E(h_{\tilde n(t)}\mid \tilde n(0) = i)\leq \exp(Ct)h_i,
    \]
    which is indeed locally integrable, confirming our claim.
\end{proof}

\end{document}